\documentclass[journal,twoside,web]{ieeecolor}
\usepackage{generic}
\usepackage{cite}
\usepackage{amsmath,amssymb,amsfonts}
\usepackage{algorithmic}
\usepackage{graphicx}
\usepackage{algorithm,algorithmic}
\usepackage{hyperref}
\hypersetup{hidelinks=true}
\usepackage{textcomp}
\usepackage{array}
\usepackage{color}

\definecolor{revisiongreen}{rgb}{0,0.5,0}

\usepackage{multirow}
\usepackage{rotating,makecell}

\newcommand{\diff}{\,\mathrm d}
\newtheorem{assumption}{Assumption}
\newtheorem{remark}{Remark}
\newtheorem{theorem}{Theorem}
\newtheorem{lemma}{Lemma}

\newtheorem{definition}{Definition}

\usepackage{subcaption}
\usepackage{soul}
\usepackage[normalem]{ulem}
\DeclareRobustCommand{\old}[1]{\begingroup
\let\oldcite\cite\let\oldref\ref\let\oldeqref\eqref
\def\cite##1{\mbox{\oldcite{##1}}}\def\ref##1{\mbox{\oldref{##1}}}
\def\eqref##1{\mbox{\oldeqref{##1}}}
\textcolor{black}{\sout{#1}}\endgroup}

\def\BibTeX{{\rm B\kern-.05em{\sc i\kern-.025em b}\kern-.08em
    T\kern-.1667em\lower.7ex\hbox{E}\kern-.125emX}}
\begin{document}
	
\title{Data-Driven Safe Output Regulation of Strict-Feedback Linear Systems with Input Delay}
\author{Zhenxu Zhao, Ji Wang, Weiyao Lan
\thanks{Z. Zhao, J. Wang, W. Lan are with  the Department of Automation, Xiamen University, Xiamen 361005, China. (E-mail addresses: zhaozhenxu@stu.xmu.edu.cn (Z. Zhao),
jiwang@xmu.edu.cn (J. Wang), wylan@xmu.edu.cn (W. Lan).)}
}

\maketitle

\begin{abstract}
This paper develops a data-driven safe control framework for linear systems possessing a known strict-feedback structure, but with most plant parameters, external disturbances, and input delay being unknown. By leveraging Koopman operator theory, we utilize Krylov dynamic mode decomposition (DMD) to extract the system dynamics from measured data, enabling the reconstruction of the system and disturbance matrices. Concurrently, the batch least-squares identification (BaLSI) method is employed to identify other unknown parameters in the input channel. Using control barrier functions (CBFs) and backstepping, we first develop a full-state safe controller. Based on this, we build an output-feedback controller by performing system identification using only the output data and actuation signals as well as constructing an observer to estimate the unmeasured plant states. The proposed approach achieves: 1)  finite-time identification of a substantial set of unknown system quantities, and 2) exponential convergence of the output state (the state furthest from the control input) to a reference trajectory while rigorously ensuring safety constraints. The effectiveness of the proposed method is demonstrated through a safe vehicle platooning application.
\end{abstract}


\section{Introduction}

\subsection{Output Regulation}
Output regulation, which achieves asymptotic tracking of reference trajectories and disturbance rejection, has long been an important topic in control theory. As the fundamental work, Internal Model Principle (IMP) was first established in \cite{FRANCIS1976457} for linear systems, which incorporates an exosystem model into the feedback controller. The work in \cite{45168} then extended this to the nonlinear dynamics by deriving the regulator equations, and a general framework was developed in \cite{huang2004nonlinear}.  In addition to ODE systems, output regulation has also been studied for PDEs \cite{outpde8,outpde9}.
A trajectory tracking result, based on a backstepping design, was presented in \cite{outpde2}, which considered an ODE with unknown input delay--a problem closely related to hyperbolic PDEs. Article \cite{outpde4} addressed in-domain point outputs, later expanding to heterodirectional systems with constant coefficients \cite{outpde5},  and  \cite{outpde6, outpde7} further generalized this to spatially varying coefficients, establishing a link between the control design and regulator equations.
In parallel,  some results have been obtained for parabolic PDEs. The work in \cite{outpde10} combined backstepping and flatness approaches to design an output-tracking boundary controller. Similarly, \cite{outpde11} addressed output regulation for boundary-controlled PDEs subject to disturbances on both the domain and the uncontrolled boundary. The scope of this research has further expanded to include cooperative output-tracking problems involving time-varying structures \cite{outpde12,outpde13}.

\subsection{Delay compensation}
Time delays, which often exist in practice, are well-known to have a significant impact on the stability of control systems.
The backstepping predictor-based delay compensation introduced in \cite{bib4,bib5}, where the delay is treated as a transport PDE, has been demonstrated to be a powerful method for handling delays. This method has been successfully extended to nonlinear and other complex delay systems in \cite{NBbook,delay10,delay11,delay12,varydelay,statedelay,Diagne1,Diagne2,karafyllis2017predictor}, constituting significant advances in delay compensation, especially for systems with long delays. A more detailed literature review on this can be found in Sec. 3.3.1 of \cite{VAZQUEZ2026112572}. For uncertain delays,
Lyapunov-based designs have been employed to develop delay-adaptive controllers in \cite{bib15,delayadaptive1,delayadaptive2,delayadaptive3}. Recently, a new delay-adaptive design \cite{delay9} was proposed for coupled PDEs, leveraging a batch least-squares identifier (BaLSI) that was originally introduced in \cite{bal1,bal2} for nonlinear ODEs, and extended to PDEs in \cite{bal3,bal4},  achieving finite-time exact identification of the uncertain delay and exponential regulation of the plant. 

\subsection{Safety-Critical Control in Complex Environments}

Safety is paramount in modern autonomous systems, where one must rigorously guarantee that system trajectories remain within a prescribed safe set. Control Barrier Functions (CBFs), introduced in \cite{ames2017control,bib25}, have proven effective for guaranteeing safety. However, practical deployment is hindered by complex real-world conditions, including time delays \cite{Zhao2025}, external disturbances \cite{XU201554}, and output-feedback constraints \cite{9867262}.  Recent work \cite{Zhao2025} integrates predictor feedback into CBF designs, using high-relative-degree CBFs \cite{bib26, nguyen2016exponential} rooted in non-overshooting techniques \cite{nonover1} to ensure safety under arbitrary-length delays.  Besides, the work in \cite{outsafe} addressed safe output regulation for hyperbolic PDE-ODE cascades. 

However, the above approaches typically assume full knowledge of the system model, whereas real-world systems often contain substantial uncertainties.  Robust CBF methods \cite{jankovic2018robust,8405547} address such uncertainties by enforcing safety under worst-case conditions, but this inevitably leads to conservative designs.  Adaptive CBFs (aCBFs) \cite{bib29}, analogous to adaptive Control Lyapunov Functions, offer an alternative but ensure safety only by restricting the state to a conservative subset of the nominal safe set. Although the BaLSI-based aCBF scheme \cite{adasafe1} reduces this conservatism by ensuring the system recovers the original safe set after a finite transient, the number of unknown parameters is limited, which is unsuitable for systems with largely unknown parameters. {In addition to these adaptive CBF approaches,
	data-driven methods have also been explored for
	safety-critical control.
	In \cite{datasafe4}, a discrete-time Koopman model
	of backup-controlled trajectories is learned from data,
	and the prediction-error bound is incorporated into
	a CBF-based safety filter for continuous-time systems.}

\subsection{Data driven with DMD and Koopman operator}
The challenge of obtaining accurate mathematical models in complex environments has driven interest in data-driven control \cite{HOU20133,data1}, where system dynamics are captured directly from data. Koopman operator theory \cite{koop1,mauroy2020koopman} is an effective tool for data-driven analysis and control, mapping nonlinear dynamics into an infinite-dimensional observable space where the evolution becomes linear. Dynamic Mode Decomposition (DMD) \cite{kutz2016dynamic,data3}  provides a practical finite-dimensional approximation of the Koopman operator directly from data \cite{KORDA2018149}, whose standard form is noise-sensitive and requires full-state data. By exploiting the intrinsic link between DMD and Krylov subspace methods \cite{dmd2}, Krylov DMD utilizes Hankel matrices to construct observable bases, enabling more robust extraction of Koopman modes and eigenvalues. A pioneering study \cite{data2} leveraged this framework to solve regulator equations solely from output data.
While data-driven control provides a powerful paradigm for stabilizing systems with unknown dynamics, establishing rigorous safety guarantees, which are crucial in practical applications, remains largely unresolved, with limited progress to date \cite{safedata1, safedata2, safedata3}.    Most existing approaches are constrained by intensive offline verification, high online computational costs, and a focus on discrete-time settings that neglect the impact of input delays. These challenges limit their applicability in practice and motivate the work presented in this paper. 

\subsection{Main contributions}
In this paper, we develop a data-driven CBF-based safe control scheme for linear strict-feedback systems with significant model uncertainties, unknown external disturbances, and unknown input delay.
The main contributions of this work are:
\begin{enumerate}
	\item Compared with data-driven output regulation results \cite{data2,10502160}, we address a more challenging case where the plant is subject to unknown input delay and safety constraints.
	{A full-rank condition  is further established in the presence of only 
		a scalar output,
		relaxing the output-dimension requirement in \cite{data2}.}
	
	\item 
	Unlike recent safe adaptive methods \cite{Zhao2025,adasafe1}, which address only a small set of unknown parameters, the present work addresses a setting in which nearly the entire plant, including disturbances and input delays, is unknown, and moreover achieves safe output regulation solely using measurable output signals.

\item {In contrast to recent data-driven safe control methods
\cite{safedata2,safedata3,datasafe4}, we develop a CBF-based
controller for continuous-time systems with unknown input
delays, incorporating Krylov DMD-based finite-time
	identification of the unknown plant and disturbance dynamics,
without  assumptions for bounds of unknown external disturbances and for the rank of data matrices.}

    \item To the best of our knowledge, this is the first result of data-driven safe control based on CBF with the Koopman operator
and Krylov DMD, even if the input delay is excluded. 
	
\end{enumerate}

\subsection{Organization}
The problem is formulated in Sec. \ref{sec2}. In Sec. \ref{secfull}, we present the identification of all unknown parameters and the data-driven safe control design under full state access.  Sec. \ref{secout} addresses the output-feedback case, where identification and safe output regulation are achieved using only input-output data. The effectiveness of the proposed design is verified in Sec. \ref{secexample} with a vehicle collision avoidance application. Finally, conclusions and future work are presented in Sec. \ref{con}.

\subsection{Notation} 
\begin{itemize}
	\item Define ${\underline{x}_i(t)}:=(x_1(t),x_2(t),\cdots,x_i(t))^\top$.
	\item  The  notation $u[t]$ denotes the profile of $u$ at a certain $t \ge 0$, i.e., $u[t] =
	u(x,t)$ for all $x\in[0, 1]$,  with norm $\|u[t]\|:=(\int_0^1u(x,t)^2\diff x)^{\frac{1}{2}}$. 
	\item The norm $|\cdot|$ denotes the usual Euclidean norm of the $n$-vector, $\| \cdot \|$ denotes the induced 2-norm of the matrix.
	\item  The symbol $\mathbb{R}$ denotes  the set of real numbers and $\mathbb{C}$ denotes  the set of complex numbers, while $\mathbb{N}$ denotes the set of all non-negative integers, and $\mathbb{R}^+ := [0, \infty)$.

	\item  The term $I_n$  denotes the ${n}\times{n}$ identity matrix and $\operatorname{eig}(A)$ denotes all the eigenvalues of matrix $A$.
 
	\item  The term $\operatorname{sgn}(u)$ represents the  signum function, i.e., $\operatorname{sgn}(u)=\left\{ \begin{array}{cc}1&u>0\\-1&u\le0\end{array}\right..$
\end{itemize}
\section{Problem Formulation}\label{sec2}

\subsection{Considered plant}
Consider the uncertain strict-feedback linear system subject to the input delay and disturbances,
\begin{align}
	\dot{X}(t)&=AX(t)+BU(t-D)+Gd(t),\label{org1}\\
	Y(t)&=CX(t),	\label{org2}
\end{align}
where the states  $X(t)=(x_1,\cdots,x_n)^\top\in\mathbb{R}^n$, the unknown delay parameter $D\in \mathbb{R}$, the  disturbance $d(t)\in \mathbb{R}^m$. To render the safe-regulation task more challenging, we consider the case with the highest relative degree. In this setting, the output matrix is $C=(1,0,\cdots,0)\in\mathbb{R}^n$, so that the measured output is $Y(t)=x_1(t)$, i.e., the state farthest from the control input.
Considering the structure of the strict-feedback systems,  the  forms of the system matrix $A$ and the input matrix are
\begin{equation}
	A = \begin{pmatrix}
  a_{1,1} & 1 & & \\
  a_{2,1} & a_{2,2} & 1 & \\
  \vdots & \vdots & \ddots & \ddots \\
  a_{n,1} & a_{n,2} & \dots & a_{n,n}
\end{pmatrix}
_{n\times n},\label{A}
	B=\begin{pmatrix}
  0 \\ 0 \\ \vdots \\ b
\end{pmatrix}_{n\times 1}.
\end{equation}
All system parameters $b$, $a_{i,j}(1\le i,j\le n)$ in \eqref{org1} are unknown. The sign of the parameter $b$ is known (without any loss of generality, $b>0$ is given). In addition, we assume that  $(A,B)$ is controllable and $(C,A)$ is observable.

Like \cite{out1,out2}, the disturbances are supposed to be generated by the following finite-dimensional exogenous model 
\begin{align}
	\dot V_d(t)=S_dV_d(t),\quad
	d(t)=P_dV_d(t),\label{modd}
\end{align}
where $V_d(t)=(v_{d,1},\cdots,v_{d,n_d})^\top\in\mathbb{R}^{n_d}$,  and $(P_d,S_d)$ is observable. The matrices $P_d\in\mathbb{R}^{m\times n_d}$ and   $S_d\in\mathbb{R}^{n_d\times n_d}$ are unknown. 
The target trajectory $r(t)\in\mathbb{R}$  is considered to be generated by the exact known exogenous model
\begin{equation}
	\dot V_r(t)=S_rV_r(t),\quad
	r(t)=P_rV_r(t),\label{modr1}
\end{equation}
where $V_r(t)=(v_{r,1},\cdots,v_{r,n_r})^\top\in\mathbb{R}^{n_r}$, $(P_r,S_r)$ is observable.
Thus, the dynamics of disturbance $d$ and the reference  $r$ are unified by a finite-dimensional exogenous model
\begin{equation}
	\dot V(t)=SV(t),  \label{modv}
\end{equation}
where $S=\operatorname{diag}(S_d,S_r)$, $V(t)=\operatorname{col}(V_d,V_r)$  and $n_v=n_d+n_r$. The matrix $S$ is partially known, and all its eigenvalues are assumed to lie on the imaginary axis. It can generate constant or trigonometric signals that are common in practical applications. 
According to \eqref{modd}--\eqref{modv}, $d, r$ can be reformulated as
\begin{equation}
	d(t)=\bar{P}_{d}V(t),\quad r(t)=\bar{P}_{r}V(t)\label{vd},
\end{equation}
where $\bar{P}_{d}=(P_d, 0)\in\mathbb{R}^{m\times n_v}$, $\bar{P}_{r}=(0,P_r)\in\mathbb{R}^{1\times n_v}$.

Model the delay $D$ in \eqref{org1} by the transport PDE, whose solution is 
\begin{equation}
	u(x,t)=\begin{cases}
		U\left(t-D+Dx\right),&\mathrm{if}\quad t-D+Dx>0,\\
		0,&\mathrm{if}\quad t-D+Dx\le 0
	\end{cases}\label{u0}
\end{equation}
for $x \in [0,1],t \in  [0,\infty)$. Recalling \eqref{vd},  we have the equivalent representation of the original plant \eqref{org1} as
\begin{align}
	\dot{X}(t)&=AX(t)+Bu(0,t)+\bar{G}V(t),\label{org11}\\
	Du_t(x,t)&=u_x(x,t),\label{org10}\\
	u(1,t)&=U(t),\label{org12}\\
	u(x,0)&=0,\label{org13}
\end{align}
where $ \bar{G}=G\bar{P}_d $ and $ \bar{G}:=(g_1;\cdots;g_n)\in\mathbb{R}^{n\times n_v}$.

{The unknown model parameters are: } 
\begin{equation}
	\Theta=\{\underbrace{A,\bar{G},S_d}_{\Theta_1}; \underbrace{D,b}_{\Theta_2}\}. \label{theta}
\end{equation}
It means that all plant parameters in the considered system are unknown except for $S_r$ that is to generate the target trajectory and thus is known. 

The following two assumptions about the bounds of unknown parameters in $\Theta_2$  are required in the control design, which are not restrictive because the bounds are arbitrary.
\begin{assumption}
	The bound of the unknown delay parameter $D$ is known and arbitrary, i.e., $D\in D_0$, where $D_0:=\{\mathfrak{D}\in\mathbb{R}:0<\underline{D}\leq \mathfrak{D}\leq\overline{D}\}$. The constants $\underline{D}$, $\overline{D}$ are arbitrary and known. \label{assumboundd}
\end{assumption}
\begin{assumption}
	The bounds of the unknown parameter $b$  are known and arbitrary, i.e., $b\in b_0$, where $b_0:=\{\beta\in\mathbb{R}:0 <\underline{b}  \le \beta \le\overline{b}\}$. The constants $\underline{b}$, $\overline{b}$ are arbitrary and known.\label{assumboundb}
\end{assumption}

Considering the safety objectives, we further impose the following boundedness assumptions on all unknown parameters in $\Theta_1$:
\begin{assumption}\label{assum:theta_bounds}
	{The bounds of unknown parameters in the system matrices $A$, $S_d$, and $\bar{G}$ are known and arbitrary. }
\end{assumption}

	Assumption \ref{assum:theta_bounds} means that the unknown parameters in $\Theta_1$ belong to a known compact set $\bar{\Theta}_1$. 
{These bounds specified in Assumptions \ref{assumboundd}--\ref{assum:theta_bounds} are employed to derive a safe controller when all unknown parameters have yet to be identified. Although these bounds may lead to conservative control, they are necessary to guarantee safety during the identification process. 
However, if a control input that maintains the states within the safe region is known for a finite duration that exceeds the lower bound $\underline{D}$ of delay, the assumption is no longer necessary.}

\subsection{Control objective}
Define the output tracking error
\begin{equation}
	e(t)=Y(t)-r(t) \label{e}
\end{equation}  and the safe set
\begin{equation}
	\mathcal{C}=\{e(t)\in\mathbb{R}:h(e(t),t)\geq0\}, \label{saferegion}
\end{equation}
where $h(e(t),t)$ satisfies the following Assumption \ref{assumh}.
\begin{assumption}\label{assumh}
The time-varying function \mbox{$h$} is \mbox{$n+1$} times continuously differentiable with respect to each of its arguments, i.e., \mbox{$e$} as well as \mbox{$t$}, and satisfies that $\frac{\partial h(e,t)}{\partial e}\neq 0$, $\forall e\in\mathbb{R}$, $t\in [0,\infty]$. Besides, $|h(e(t),t)|<\infty \Rightarrow |e(t)|<\infty$ and $\lim_{t\rightarrow\infty}h(e(t),t)=0 \Rightarrow e(t)=0$.
\end{assumption}

{Under Assumption \ref{assumh}, the partial derivative
	$\partial h/\partial e$ is continuous and nonzero, which implies
	that $h(e,t)$ is strictly monotone in $e$ for each fixed $t$.
	The safe set $\mathcal C$ therefore describes a one-sided
	constraint on the output tracking error, such as a minimum
	vehicle-spacing requirement.}

The control objective is to design a controller to  ensure the exponential convergence to zero of the error state $e(t)$, i.e.,
\begin{equation}
	\lim_{t\rightarrow\infty}	e(t)=Y(t)-r(t)=0,
\end{equation}
and to keep the error  state $e(t)$ in the safe set $\mathcal{C}$, in the situation that almost all plant parameters are unknown.

\section{Control design with full state access}\label{secfull}
In this section, we first design a nominal safe controller using a three-step backstepping transformation, assuming full model knowledge. Then, addressing the scenario where nearly all plant parameters \eqref{theta} are unknown,  a data-driven approach that leverages the Koopman operator and Krylov DMD, along with a BaLSI identifier, is utilized to extract the underlying dynamics of the plant \eqref{org1}. Finally, we construct a safe controller that guarantees safety enforcement and exponential regulation, even in the presence of significant model uncertainties.

\subsection{Nominal safe controller design}\label{secnom}
We proceed with the nominal control design following \cite{outsafe}.
\subsubsection{First transformation}\label{secfistran}
Firstly, we introduce the linear coordinate transformation to  convert the original system \eqref{org11} and  the exogenous system \eqref{modv}  into a vector chain of  integrators:
\begin{equation}
	Z(t)=T_zX(t)+T_vV(t),\label{Z}
\end{equation}
where $Z(t)=(z_{1},\cdots,z_{n})^\top$, and the matrices $T_z\in\mathbb{R}^{n\times n}$ as well as $T_v\in\mathbb{R}^{n\times n_v}$ are 
$
	T_z = \begin{pmatrix}
  1 & & & \\
 \varrho_{1,1} & 1 & & \\
  \vdots & \ddots & \ddots & \\
  \varrho_{n-1,1} & \dots & \varrho_{n-1,n-1} & 1
\end{pmatrix}
$,
$
	T_v=-\begin{pmatrix}\bar{P}_r;\sigma_1+\bar{P}_rS^1;\cdots;\sigma_{n-1}+\bar{P}_rS^{n-1}\end{pmatrix}.
$
The constants $\varrho_{i,j}$ in $T_z$ are defined as 
$
\varrho_{1,1}=a_{1,1},\varrho_{2,1}=a_{2,1}+\varrho_{1,1}a_{1,1},\varrho_{2,2}=a_{2,2}+\varrho_{1,1},
$
for $i=3,\cdots,n $, as
$\varrho_{i,1}=a_{i,1}+\sum_{j=1}^{i-1}\varrho_{i-1,j}a_{j,1},\varrho_{i,\iota}=a_{i,\iota}+\varrho_{i-1,\iota-1}+\sum_{j=\iota}^{i-1}\varrho_{i-1,j}a_{j,\iota},\forall\iota=2,\cdots,i-1,\varrho_{i,i}=a_{i,i}+\varrho_{i-1,i-1}$,
and the vectors $\sigma_i$  are defined as 
$\sigma_{0}=0,\sigma_{i}=-\sum_{j=1}^{i-1}\varrho_{i-1,j}g_{j}-g_{i}+\sigma_{i-1}S,  i=1,\cdots,n.
$
With this transformation, we obtain
\begin{equation}
	\dot{Z}(t)=	A_zZ(t)+B\left(u(0,t)+KX(t)+G_0V(t)\right),
\end{equation}
where 
\begin{equation}
	A_z = \begin{pmatrix}
  0 & 1 & & \\
    & 0 & \ddots & \\
    & & \ddots & 1 \\
    & & & 0
\end{pmatrix},K^\top=\frac{1}{b}\times\begin{pmatrix} \varrho_{n,1}\\ \vdots\\ \varrho_{n,n}\end{pmatrix},\label{K}
\end{equation}
and $G_0=-\frac{1}{b}(\sigma_n+\bar{P}_rS^n)$. Details of the transformation can be found in [\citen{outsafe}, Appendix 0.1].
To compensate for the delay, it is necessary to construct the prediction of system states. The prediction of $X(t)$ can be expressed as
\begin{align}
	X(t+\tau)=e^{A\tau}X(t)+\int_{0}^{\frac{\tau}{D}}De^{DA(\frac{\tau}{D}-y)}Bu(y,t)\diff y\notag\\
	+\int_{0}^{\frac{\tau}{D}}De^{DA(\frac{\tau}{D}-y)}\bar{G}e^{DSy}V(t)\diff y, \tau\in[0,D].\label{predx}
\end{align}

\subsubsection{Second transformations}
Regarding the safe barrier function $h$, we introduce the second transformation
\begin{align}
	&h_1(z_1(t),t)=h(e(t),t)+\varsigma(t),\label{h1}\\
	&h_i(\underline{z}_i,t)=\sum_{j=1}^{i-1}\frac{\partial h_{i-1}}{\partial z_j}z_{j+1}+\frac{\partial h_{i-1}}{\partial t}+k_{i-1}h_{i-1},\label{hi}
\end{align} 		
for $i=2,\cdots,n$. Like \cite{outsafe}, $\varsigma(t)$ is introduced to handle the initially unsafe case. That is, if the error state $e(t)$ falls outside of the safe region at the initial regulation time $t=D$, i.e., $h(e(D),D) \le 0$, the function $\varsigma(t)$ plays a role in rescuing to the defined safe region  \eqref{saferegion}. The function $\varsigma(t)$ is given by
\begin{equation}
	\varsigma(t)=\begin{cases}
		\begin{cases}
			(-h(e(D),D)+\epsilon)e^{\frac{1}{\bar{t}^2}-\frac{1}{(t-D-\bar{t})^2}},0\le t <D+\bar{t}\\
			0, \phantom{=================}t\geq D+\bar{t}
		\end{cases}\\
		\phantom{===========,,}if\quad h(e(D),D)\leq0.\\
		0,\quad t\geq 0,\quad \phantom{=====}if\quad h(e(D),D)>0
	\end{cases}\label{sigma}
\end{equation}
where $\epsilon$ is an arbitrarily positive design parameter, and the positive design parameter  $\bar{t}$ is the prescribed time for recovery to the safe region, which can be arbitrarily assigned by users.  {With the known model and delay in this nominal design, $h(e(D),D)$ can be calculated at $t=0$ from $X(0)$, $V(0)$ using \eqref{predx}, \eqref{org1}, \eqref{modv}.} It is compactly written as
\begin{equation}
	h(e(D),D) := \mathcal{P}_0(\chi(0),\Theta),\label{p0}
\end{equation}
where 
\begin{equation}
	\chi(t):=(V(t), X(t)). \label{chi}
\end{equation}

It is obtained from \eqref{sigma} that $\varsigma(t)$ is continuously differentiable of all orders. Therefore, together with Assumption \ref{assumh}, $h(e(t),t)$
is $n$ times differentiable with respect to all its arguments $e(t)$ and $t$.
Taking the partial derivative for $h(e(t),t)$ with  respect to the
variable $e(t)$, we have
\begin{equation}
	\frac{\partial h(e(t),t)}{\partial e(t)}:=\vartheta(e(t),t).\label{var}
\end{equation}
Given the linearity of the system, it follows from \eqref{h1} and $e(t)=z_1(t)$ that $\frac{\partial h_n}{\partial z_n}=\frac{\partial h_{n-1}}{\partial z_{n-1}}=\frac{\partial h_{n-2}}{\partial z_{n-2}}=\cdots=\frac{\partial h_1}{\partial z_1}=\frac{\partial h}{\partial e}=\vartheta$.	 
Thus, applying the above transformations \eqref{h1}--\eqref{sigma}, defining $H=(h_1,\cdots,h_n)^\top$, we have
\begin{align}
	\dot{H}(t)&=A_{\mathrm{h}}H(t)
	+B\Big( f(\underline{z}_n(t),t)\notag\\	&+\vartheta\big(u(0,t)+KX(t)+G_0V(t)\big)\Big),\label{doth}
\end{align}
where
$
	A_\mathrm{h} = \begin{pmatrix}
  -k_1 & 1 & & \\
   & -k_2 & \ddots & \\
   & & \ddots & 1 \\
   & & & -k_n
\end{pmatrix},
\label{ah}
$
and
$
	f(\underline{z}_n(t),t)=\frac{1}{b}\left(\sum_{j=1}^{n-1}\frac{\partial h_n}{\partial z_j}z_{j+1}+\frac{\partial h_n}{\partial t}
	+k_nh_n\right).
$

\subsubsection{Third Transformation and controller}
 To remove the dependence on the
signal model state and convert the intermediate system \eqref{doth}, \eqref{org10},\eqref{org12} into an exponentially stable form, we introduce the following backstepping transformation: 
\begin{align}
	&w(x,t)=u(x,t)+\int_0^x{DKe^{DA(x-y)}B}u (y,t)\mathrm{~d}y-\bar{\gamma}(x)V(t)\notag\\
	&+Ke^{DAx}X(t)+\frac{f\left(\underline{z}_n(t+Dx),t+Dx\right)}{\vartheta(z_1(t+Dx),t+Dx)}\label{wxt}
\end{align}
where $\bar{\gamma}(x)$ is the regulator equation's solution:
\begin{equation}	
	\bar{\gamma}(x)=-\int_{0}^{x}DKe^{DAy}\bar{G}e^{DS(x-y)}\diff y
	+\frac{1}{b}(\sigma_n+\bar{P}_rS^{n})e^{DSx}.\label{bargam}
\end{equation}
The inverse transformation of \eqref{wxt} is given by:
\begin{align}
	u(x,t)=&w(x,t)+\int_0^x\bar{q}(x,y)w(y,t)\mathrm{~d}y+\gamma_1(x)X(t)\notag\\
	&+{\bar{\gamma}_1(x)}V(t)+\psi_1(x,t).\label{inverw}
\end{align}
The inverse transformation \eqref{inverw} is well-defined and its explicit calculation is given in Appendix \ref{pdetrans}.

Applying \eqref{wxt},  we convert the intermediate system  \eqref{doth}, \eqref{org10}, \eqref{org12}  into the target system
\begin{align}
	\dot{H}(t)&=A_{\mathrm{h}}H(t)+B\vartheta w(0,t),\label{tar1}\\
	Dw_t(x,t)&=w_x(x,t)\label{tar2}\\
	w(1,t)&=0,	\label{tar3}	
\end{align} 
with choosing the nominal controller
\begin{align}
	&U(t)=-\int_{0}^{1}DKe^{DA(1-y)}Bu(y,t)\diff y-Ke^{DA}X(t)\notag\\
	&-\big(\int_{0}^{1}DKe^{DAy}\bar{G}e^{DS(1-y)}\diff y+G_0e^{DS}\big)V(t)+\psi(1,t)\notag\\
	&:=\mathcal{U}(\chi(t);\Theta), \label{u}
\end{align}
where
\begin{align}
	\psi(1,t)&=-\frac{f\left(\underline{z}_n(t+D),t+D\right)}{\vartheta(z_1(t+D),t+D)}:=\bar{\mathcal{P}}(Z(t), u[t],V(t),t,D)\notag\\
	&=\bar{\mathcal{P}}(T_{z}X(t)+T_{v}V(t),u[t],V(t),t,D)\label{p}
\end{align}
and where $\Theta$ defined in \eqref{theta}, and $\chi(t)$ defined in \eqref{chi} are included in \eqref{u} to emphasize the controller's dependence on the signals  $X(t)$, $V(t)$, and parameters $\Theta$.

\begin{assumption}\label{assumgrowth}
	The barrier function $h$ ensures that there exists a positive $\xi_z$ such that
	\begin{align}
		\left|\bar{\mathcal P}\right|\le\xi_z(1+|X(t)|+|V(t)|+\|u[t]\|),\label{pbound}
	\end{align}
	 where $\bar{\mathcal P}$  \eqref{p} is determined by $h$ via \eqref{h1}--\eqref{var}.
\end{assumption}

The above imposes a growth condition on the term $\bar{\mathcal{P}}$ to establish boundedness during parameter identification. It is not required for the nominal design,   because the target system \eqref{tar1}--\eqref{tar3} under the exact model has zero boundary input.

\subsubsection{Selection of the design parameters}
The selection of the gain parameters $k_i$ in \eqref{hi} is essential for both regulation and safety. The positive gains $k_i$, $i=1,\cdots,n$, are chosen to satisfy:
\begin{equation}
	k_i\ge\max\{0,\hat{k}_i\},i=1,\cdots,n-1, \qquad k_n>0, \label{k}
\end{equation}
where
\begin{equation}
	\hat{k}_i\left(\chi(0);\Theta\right)=\frac{-1}{h_i\left(\underline{z}_i\left(D\right), D\right)}\left(\sum_{j=1}^i \frac{\partial h_i}{\partial z_j} z_{j+1}\left(D\right)+\frac{\partial h_i}{\partial t}\right).\label{hk}
\end{equation}
The vector $Z(D)$ can be precisely determined a priori at the initial time $t=0$ by evaluating the state predictor \eqref{predx}.

\subsubsection{Result of the nominal controller}
We now present the following two lemmas to summarize the results achieved by the nominal controller.
\begin{lemma}\label{lemmasafe2} 
	The high-relative-degree CBFs $h_i(t), i=1,\cdots,n$ are nonnegative under the selection of design parameters \eqref{k}, \eqref{hk}, i.e., $h_i(t)\ge 0, i=1,\cdots,n$, for time $t\ge D$.
\end{lemma}
\begin{proof}
	The proof is shown in Appendix \ref{appcbf}.
\end{proof}

\begin{lemma}\label{lembound}
	The tracking error $e(t)$  \eqref{e} in the closed-loop system \eqref{org1}, \eqref{org2} with the nominal controller \eqref{u} converges to zero, and all system states remain bounded for any initial data $X(0)\in\mathbb{R}^n$, $V(0)\in\mathbb{R}^{n_v}$, and for any safety constraint $h(e,t)$ satisfying Assumption \ref{assumh}, provided that the design parameters $k_1,\cdots,k_n$ are chosen satisfying \eqref{k}, \eqref{hk}. 
\end{lemma}
\begin{proof}
	The proof is shown in Appendix \ref{appbound}.
\end{proof} 
Based on the above nominal control design, the remainder of this section turn to safe data-driven design for unknow parameters.

\subsection[Extract system dynamics $\Theta_1$ using DMD]{Extract system  dynamics $\Theta_1$ using DMD} \label{secdmd}

 The presence of the input delay makes the plant behave as an autonomous system during the initial delay interval $t\in[0,D]$. This allows us to use DMD to extract important dynamic characteristics $\Theta_1=\{A,\bar{G},S_d\}$ without corruption from the external control input. We first introduce the piecewise constant function $\hat{\Theta}_1(t)$ to denote the estimated parameters of $\Theta_1$ with arbitrary initial value $\hat{\Theta}_1(0)$ satisfying Assumption \ref{assum:theta_bounds}.
 
\subsubsection{The extended system with its Koopman operator}
Recalling \eqref{org1}, \eqref{org2}, \eqref{modd}, during the initial period,  the plant can be modeled as an autonomous extended system
\begin{equation}\label{tilA}
	\dot{\tilde{X}}(t)=\underbrace{\begin{pmatrix}
			S_d&0\\
			GP_d&A
	\end{pmatrix}}_{\tilde{A}}\tilde{X}(t)
	,\quad t\in [0, D]
\end{equation}
where $ \tilde{X}(t)=(V_d,X)^\top\in\mathbb{R}^{\tilde{n}}$ and ${\tilde{n}}=n+n_d$. Besides, the extended matrix $\tilde{A}\in\mathbb{R}^{\tilde{n}\times\tilde{n}}$ satisfies the following assumption:
\begin{assumption}\label{assumeigen}
	The eigenvalues of the matrix $\tilde{A} $ are simple. 
\end{assumption}

By stating that the eigenvalues are simple, we impose that the algebraic multiplicity of each eigenvalue is one such that the matrix $\tilde{A}$ is diagonalizable. This assumption contributes to the clarity of the presentation. The proposed method can also be extended to the case of cyclic system matrices.

As shown in \cite{data2}, the Koopman operator $\mathcal{K}(t)$ describes the evolution of the  system \eqref{tilA} through observables $\phi(\tilde{X})$ according to
$
\mathcal{K}(t)\phi(\tilde{X})=\phi(e^{\tilde{A}t}\tilde{X}).\label{koopobe}
$
For the identity observable $\phi(\tilde{X})=\tilde{X}$ for the linear system, the eigenfunctions satisfy $\mathcal{K}(t)\tilde\varphi_i(\tilde{X})=e^{\tilde{\lambda}_it}\tilde\varphi_i(\tilde{X})$, $i=1,\cdots,\tilde{n}$, with $\tilde{\lambda}_i\in\mathbb{C}$.
These eigenfunctions take the specific form \cite{data2}:
\begin{equation}
\tilde\varphi_i(\tilde{X})=\langle\tilde{X},\tilde{\omega}_i\rangle, \label{eigenfun}
\end{equation}
where $\tilde{\omega}_i\in\mathbb{C}^{\tilde{n}}$ are left eigenvectors of $\tilde{A}$, and $\langle\cdot,\cdot\rangle$ denotes the usual inner product. The state admits the decomposition 
\begin{equation}\tilde{X}=\sum_{i=1}^{\tilde{n}}\tilde{\nu}_i\tilde\varphi_i(\tilde{X})=\tilde{V}\underline{\tilde\varphi}(\tilde{X}), \label{tilx}
\end{equation}  where the Koopman modes $\tilde{\nu}_i$ are the right eigenvectors of $\tilde{A}$ associated with $\tilde{\lambda}_i$, with $\tilde{V}=(\tilde{\nu}_1,\cdots,\tilde{\nu}_{\tilde{n}})\in\mathbb{C}^{\tilde{n}\times\tilde{n}}$ and $\underline{\tilde\varphi}(\tilde{X})=[\tilde\varphi_1(\tilde{X}),\cdots,\tilde\varphi_{\tilde{n}}(\tilde{X})]^\top$.  
Especially,  the block structure of $\tilde{A}$ gives
\begin{equation}
	\tilde{\nu}_i = 
	\begin{cases}
		(\nu_{d,i}, \, \nu_{x,i})^\top, & i = 1, \dots, n_d \\
		(0, \, \nu_{x,i})^\top,         & i = n_d+1, \dots, \tilde{n}
	\end{cases}\label{nu}
\end{equation}
where $S_d{\nu}_{d,i}=\tilde{\lambda}_i{\nu}_{d,i}$ with
\begin{equation}
	{\nu}_{x,i}=(\tilde{\lambda}_iI-A)^{-1}GP_d{\nu}_{d,i}\label{nux}
\end{equation}
for $i=1,\cdots,n_d $, and
$
A{\nu}_{x,i}=\tilde{\lambda}_i{\nu}_{x,i}\label{nux2}
$ 
for $i=n_d+1,\cdots,\tilde{n}$.

\subsubsection{Krylov DMD}
To identify the underlying dynamics $\Theta_1$, we collect data at discrete time intervals with a constant sampling period $T_d > 0$  satisfying the following assumption.

\begin{assumption}\label{assumtd}
	The sampling time $T_d$ satisfies
	$e^{\tilde{\lambda}_i T_d} \neq \mathrm{e}^{\tilde{\lambda}_j T_d}, i \neq j ,\forall \tilde{\lambda}_i, \tilde{\lambda}_j \in \operatorname{eig}(A) \cup \operatorname{eig}\left(S_d\right)$.
\end{assumption}

{This assumption prevents distinct continuous-time modes
from becoming indistinguishable after sampling,
thereby preserving the observability and controllability
of the continuous-time system.
For example, two oscillatory modes whose phase increments
differ by an integer multiple of $2\pi$ become
indistinguishable in the sampled data.}
Although the matrices $A$ and $S_d$ are unknown,
the set of inadmissible sampling periods has measure zero
for a fixed system, which renders the condition non-restrictive \cite{data2}.

 This process generates a sequence of data points, which are arranged into the data matrix 
\begin{equation}
	\tilde{\mathcal{X}}=\begin{pmatrix}\tilde{X}(0)\quad \tilde{X}(1)\quad \cdots\quad \tilde{X}(2\tilde{n}-1)\end{pmatrix}\in\mathbb{R}^{\tilde{n}\times2\tilde{n}},\label{snap}
\end{equation}
where each column $\tilde{X}(k) = \tilde{X}(kT_d)$, $k \in \mathbb{N}$, is referred to as  a snapshot \cite{data3}.

\begin{remark}\label{remsam}
	We are making use of the delay to identify the uncertainties. The uncertain $D$ requires us to limit data collection in the interval $[0, \underline{D}]$, considering the known lower bound in Assumption \ref{assumboundd}. Consequently, the sampling time must satisfy $0 < T_d \leq \frac{\underline{D}}{2\tilde{n}-1}$. {If $\underline{D}$ is too small, the data collection period can be extended by applying zero input.}
	
\end{remark}

From the sequence of snapshots \eqref{snap}, we construct the Hankel matrix $H_{\tilde{n}}(\tilde{X})$:
\begin{equation}
	H_{\tilde{n}}(\tilde{X})=\begin{pmatrix}\tilde{X}(0)&\tilde{X}(1)&\ldots&\tilde{X}(\tilde{n}-1)\\\tilde{X}(1)&\tilde{X}(2)&\ldots&\tilde{X}(\tilde{n})\\\vdots&\vdots&\vdots&\vdots\\\tilde{X}(\tilde{n}-1)&\tilde{X}(\tilde{n})&\ldots&\tilde{X}(2\tilde{n}-2)\end{pmatrix}\in\mathbb{R}^{\tilde{n}^2\times\tilde{n}}.
\end{equation}
This matrix is known to admit the factorization 
\begin{equation}
	H_{\tilde{n}}(\tilde{X})=Q_oQ_c, \label{qoc}
\end{equation}
where $Q_o$ is the observability matrix and $Q_c$ is the controllability matrix. They are defined as
\begin{equation}
    Q_o = \begin{pmatrix} I_{\tilde{n}} & {\tilde{A}_d}^\top & \cdots & ({\tilde{A}_d^{\tilde{n}-1}})^\top \end{pmatrix}^\top\in\mathbb{R}^{\tilde{n}^2\times\tilde{n}}\label{qo},
\end{equation}
\begin{equation}
    Q_c = \begin{pmatrix} \tilde{X}(0) & \tilde{A}_d\tilde{X}(0) & \cdots & \tilde{A}_d^{\tilde{n}-1}\tilde{X}(0) \end{pmatrix}\in\mathbb{R}^{\tilde{n}\times\tilde{n}}. \label{qc}
\end{equation}
where 
$
	\tilde{A}_d=e^{\tilde{A}T_d}
$.

 We now introduce several lemmas to describe the properties of the Hankel matrix $H_{\tilde{n}}(\tilde{X})$ and to show how system characteristics can be extracted from it. 
\begin{lemma}\label{lemfull1}
	The Hankel matrix $	H_{\tilde{n}}(\tilde{X}) $
	satisfies
	$
		\operatorname{rank}	H_{\tilde{n}}(\tilde{X}) =\tilde{n}
	$
	for almost all $\tilde{X}(0)\in\mathbb{R}^{\tilde{n}}$. 
\end{lemma}
\begin{proof}
The proof is shown in Appendix \ref{appfull1}.
\end{proof}

The full-rank property established in Lemma \ref{lemfull1} enables the formulation of the Krylov DMD method \cite{dmd1,dmd2,data2}, which we describe as follows.
\begin{lemma}
	If $\operatorname{rank}H_{\tilde{n}}(\tilde{X})=\tilde{n}$,
	\begin{equation}
		\mathcal{K}(T_d)H_{\tilde{n}}(\tilde{X})=H_{\tilde{n}}(\tilde{X})\underbrace{\begin{pmatrix}
  0 & & & -f_0 \\
  1 & 0 & & -f_1 \\
    & \ddots & \ddots & \vdots \\
    & & 1 & -f_{\tilde{n}-1}
\end{pmatrix}}_{F\in\mathbb{R}^{\tilde{n}\times\tilde{n}}}\label{kh}
	\end{equation}
	with $
	f=\mathrm{col}(f_{0},\ldots,f_{\tilde{n}-1})=-(H_{\tilde{n}}^{\top}(\tilde{X})H_{\tilde{n}}(\tilde{X}))^{-1}H_{\tilde{n}}^{\top}(\tilde{X})\allowbreak\mathrm{col}(\tilde{X}(\tilde{n}),\ldots,\tilde{X}(2\tilde{n}-1))$
	holds.
\end{lemma}

\begin{proof}
	Its proof can be found in [\citen{data2}, Lemma 3].
\end{proof}
The next step is crucial, which derives the Koopman elements from the companion matrix $F$.
\begin{lemma}\label{lemeigen}
	If $\operatorname{rank}H_{\tilde{n}}(\tilde{X})=\tilde{n}$, and  letting $\hat{\nu}_i,i=1,\cdots,\tilde{n}$  be   the eigenvectors of $F$   corresponding to the eigenvalues $\hat{\lambda}_i$, the Koopman eigenvalues can be computed from
$
		\tilde{\lambda}_i=\frac{\ln\hat{\lambda}_i}{T_d},\quad i=1,\ldots,\tilde{n}
$
	including the eigenvalues $\tilde{\lambda}_i, i =1,\cdots,n_d$, of the disturbance model $S_d$ in \eqref{modd} and $\tilde{\lambda}_i, i =n_d+1,\cdots,\tilde{n}$ of the  system  matrix $A$ in \eqref{org1}. The  Koopman modes $\tilde{\nu}_i$  corresponding to the Koopman eigenvalues $\tilde{\lambda}_i$  are
$
		\tilde{\nu}_{i}=(I_{\tilde{n}}\quad 0)H_{\tilde{n}}(\tilde{X})\hat{\nu}_i,i=1,\cdots,\tilde{n}.
$
\end{lemma}

\begin{proof}
	The proof is shown in Appendix \ref{appf}.
\end{proof}

\begin{lemma}\label{lemdmd}
If $\operatorname{rank}H_{\tilde{n}}(\tilde{X})=\tilde{n}$, the exact identification  of the unknown system dynamics ${\Theta}_1$ is achieved at $t=\underline{D}$, that is $\hat{\Theta}_1(\underline{D})=\{\mathcal{A},\mathcal{S}_d,\bar{\mathcal{G}}\}={\Theta}_1$, where
\begin{align}
    \mathcal{A} &= (\mathbf{0}\quad I_n)\,\, \tilde{\mathcal{A}}\,\, (\mathbf{0}\quad I_n)^\top, \label{matha} \\
    \mathcal{S}_d &= (I_{n_d}\quad \mathbf{0})\,\, \tilde{\mathcal{A}}\,\, (I_{n_d}\quad \mathbf{0})^\top,  \label{mathsd}\\
    \bar{\mathcal{G}} &= \Big( (\mathbf{0}\quad I_n)\,\, \tilde{\mathcal{A}}\,\, (I_{n_d}\quad \mathbf{0})^\top\quad \mathbf{0} \Big) \label{mathg},
\end{align}
    with identifying $\tilde A$ by
\begin{equation}
\tilde{\mathcal{A}}=\tilde{V}\tilde{\Lambda}\tilde{V}^{-1}.\label{eq:tA}
\end{equation}
The matrices $\tilde{V}$ and $\tilde{\Lambda}=\operatorname{diag}(\tilde{\lambda}_1,\cdots,\tilde{\lambda}_{\tilde n})$ are obtained from Lemma \ref{lemeigen}.
\end{lemma}
\begin{proof}
The proof is shown in Appendix \ref{appdmd}.
\end{proof}

\subsection{Identification of parameters $\Theta_2$}\label{secada}

Following \cite{delay9}, we use a batch of actuation and plant data collected over a duration to build the following identifier: 
\begin{align}
	&\hat{\Theta}_2(t_{i+1})=\arg\min\Big\{|\ell-\hat{\Theta}_2(t_{i})|^2: \ell\in\bar{\Theta}_2,\notag\\
	&\bar{F}_n(t_{i+1},\mu_{i+1})=\bar{Q}_n(t_{i+1},\mu_{i+1})\ell,n=1,2,\cdots\Big\},\label{est1}
\end{align}
to estimate the unknown parameters $\Theta_2 = \{D, b\}$, i.e., the unknown delay and the parameter in the input matrix $B$, where $\hat{\Theta}_2(t_{i})=(\hat{D}(t_{i}),\hat{b}(t_{i})) $ and the set 
\begin{equation}
	\bar{\Theta}_2:=\{\ell\in \mathbb{R}^2:\underline{D}\leq \ell_D\leq\overline{D},\underline{b}\leq \ell_b\leq\overline{b}\} \label{thetabound}
\end{equation}
uses the known bounds given in Assumption \ref{assumboundd}, \ref{assumboundb}.  The parameter estimates \eqref{est1} are updated at the discrete sequence of time instants ${\{t_i\ge0\}}_{i=0}^\infty,i=0,1,2,\cdots$, by employing the fixed-time triggering mechanism:
$
	t_{i+1}=t_i+T_a, \label{tritime}
$
where $T_a> 0$ is a design parameter satisfying 
$
	t_1=t_0+T_a\ge\underline{D}.\label{eq:ta}
$
The above condition ensures that updates commence only after the DMD identification of $a_{n,i}$ completes at $t = \underline{D}$, as these parameters are required to compute \eqref{est1} (specifically $f_b(t)$).

Each new estimate at $t_{i+1}$ is computed using a batch of historical plant data collected over the time window $[\mu_{i+1}, t_{i+1}]$. The start of this window $\mu_{i+1}$ is defined as \begin{equation}
	\mu_{i+1}:=\min\{t_g:g\in{0,\cdots,i},t_g\ge t_{i+1}-\tilde{N}{T_a}\},\label{mu}
\end{equation}
where $\tilde{N}\ge 1$ is a free integer parameter.
The matrices $\bar{F}_{{n}}$ and $\bar{Q}_{{n}}$ used in the estimator \eqref{est1} are defined as:
$
	\bar{F}_{{n}}=\left(F_{{n}}, F_b\right)^\top$, $\bar{Q}_{{n}}=\begin{pmatrix}
		Q_n  & 0 \\
		0 & Q_{b} 
	\end{pmatrix}$.
Their components are computed by integrating signals over the data window $[\mu_{i+1}, t_{i+1}]$:
\begin{align}
	F_n=\int_{\mu_{i+1}}^{t_{i+1}}f_n(t)q_n(t)\diff t,\quad
	Q_n=\int_{\mu_{i+1}}^{t_{i+1}}q_n(t)^2\diff t	,\label{gn}
\end{align}
$F_b=\int_{\mu_{i+1}}^{t_{i+1}}f_b(t)q_b(t)\diff t$, $Q_b=\allowdisplaybreaks\int_{\mu_{i+1}}^{t_{i+1}}q_b(t)^2\diff t$, where 
$q_n(t)=-\int_{0}^{1}\sin(x\pi n)u(x,t)\diff x$, $f_n(t)=\pi  n\int_{0}^{t}\int_{0}^{1}\cos(x\pi n)u(x,\tau)\diff x\diff\tau,f_b(t)=x_n(t)-x_n(0)-\int_{0}^{t}\sum_{i=1}^{n}a_{n,i}x_i(\tau) \diff\tau-\int_{0}^{t}g_nV(\tau)\diff \tau,q_b(t)=\int_{0}^{t}u(0,\tau)\diff \tau$.
The estimator \eqref{est1} follows from the regression identities $f_n=Dq_n$, $f_b=bq_b$. Its derivation is given in Appendix \ref{appbalsi}.

\begin{lemma}\label{lemada}
	The finite-time exact identification of the unknown
	parameters $\Theta_2$ is achieved, that is, there exists a finite time $t_f$
	defined by 
	\begin{equation}
		t_f=\min\{t_i:\exists t\in [0,t_i),u(0,t)\neq0\}, \label{tf}
	\end{equation}
	such that $\hat{\Theta}_2(t_f)\equiv\Theta_2, \forall t\ge t_f$.
\end{lemma}
\begin{proof}
	{The proof is shown in Appendix \ref{appbalsi}.}
\end{proof}

\begin{remark}\label{remactuationmeasure}
{In this paper, the expression $u(x,t)=U(t-D+Dx)$ is only to show the relationship between the delay and the transport PDE. The PDE state $u(x,t)$ is assumed to be directly measurable as the actuator state for the ODE plant. This is feasible in some cases where the delay is the result of a physical process Ch.8 in \cite{bib4}. The quantities $q_n$, $f_n$, and $q_b$ are computed by integrating the measured actuator states $u(x,t)$.
One possible approach to relaxing the requirement for actuator states is to develop a PDE observer to estimate the actuator states, like Ch.8 in \cite{bib4}. This would require an additional analysis of the effect of the estimation errors on parameter identification and safety enforcement, and also the obtained global result will be reduced to the local result.}
\end{remark}

\subsection{Data-Driven Safe Control}\label{seccon2}
 Replacing the unknown parameters by their real-time estimates $\hat{\Theta}(t)=(\hat{\Theta}_1(t),\hat{\Theta}_2(t))$ obtained from Secs. \ref{secdmd}, \ref{secada} in the nominal safe control law \eqref{u}, we have 
\begin{equation}
	U_d(t):=\mathcal{U}(\chi(t);\hat\Theta(t)),\label{ud}
\end{equation}
where $\hat{\Theta}(t)$ is  a piecewise-constant vector recalling the designs in Secs. \ref{secdmd}, \ref{secada}:
\begin{equation}
	\hat{\Theta}(t)=\begin{cases}
		(\hat{\Theta}_1(0),\hat{\Theta}_2(0)),&t\in[0,\underline{D})\\
		(\hat{\Theta}_1(\underline{D}),\hat{\Theta}_2(0)),&t\in[\underline{D},t_f)\\
		(\hat{\Theta}_1(\underline{D}),\hat{\Theta}_2(t_f)),&t\in[t_f,\infty)
	\end{cases}\label{thetahat}
\end{equation}
where $\hat{\Theta}_1(0)$, $\hat{\Theta}_2(0)$ are the initial estimates, and $\hat\Theta_1(\underline D)=\{\mathcal{A},\mathcal{S}_d,\bar{\mathcal{G}}\}$, which are obtained by  \eqref{matha}--\eqref{mathg}, is the exact identification of $\Theta_1$ according to Lemma \ref{lemdmd},  and $\hat\Theta_2(t_f)$ is the exact identification of $\Theta_2$ according to Lemma \ref{lemada}.

 Different from the nominal control $U$, here $U_d(t)$ is potentially unsafe due to the identification error of $\Theta$ for $t\in[0,t_f)$. Especially, the safety mechanism designed about \eqref{sigma},  \eqref{hk} relies on the exact value of $h(e(D),D)$ recalling \eqref{p0}, which is inaccessible in the case of the uncertain plant. We instead compute its lower and upper bounds, $\underline{h}$ and $\overline{h}$, over the parameter uncertainty sets $\bar{\Theta}_1,\bar{\Theta}_2$ defined in Assumptions \ref{assum:theta_bounds} and \eqref{thetabound}, as $\overline{h}=\max_{\theta_1\in\bar{\Theta}_1,\theta_2\in\bar{\Theta}_2}\mathcal{P}_0(\chi(0);\theta_1,\theta_2),$
	\begin{align}
		\underline{h}=\min_{\theta_1\in\bar{\Theta}_1,\theta_2\in\bar{\Theta}_2}\mathcal{P}_0(\chi(0);\theta_1,\theta_2).\label{uh}
	\end{align}
For \eqref{sigma}, the sign of $h(e(D),D)$ should be known. Therefore, we impose the following assumption to ensure the sign is unambiguous:
	\begin{assumption}
	The lower and upper bounds $\underline{h}$ and $\overline{h}$ have the same sign, i.e., they satisfy either $\overline{h} \le 0$ or $\underline{h} > 0$.\label{assumsign}
	\end{assumption}

This assumption determines whether $h(e(D),D)$ is nonpositive  from the known bounds. When Assumption \ref{assumsign} does not hold, the sign of $h(e(D),D)$ cannot be reliably inferred in the presence of uncertainties. Therefore, the system has to be conservatively treated as unsafe at that instant, i.e., $h(e(D),D)\le 0$ is adopted in the subsequent control design, which will lead to rigorous safety guarantees established starting from $t=\bar D+\bar t$ with a positive $\bar t$ assigned by users.

	Then the modified function $\varsigma$ \eqref{sigma} is introduced as
	\begin{equation}
		\varsigma(t)=\begin{cases}
			\begin{cases}
				(-\underline{h}+\epsilon)e^{\frac{1}{\bar{t}^2}-\frac{1}{(t-\overline{D}-\bar{t})^2}},&0\le t<\overline{D}+\bar{t},\\
				0,&t\geq \overline{D} +\bar{t}
			\end{cases}\\
			\phantom{==========,}    if\quad \overline{h} \leq0.\\
			0,\quad t\geq 0,\quad \phantom{===,} if\quad \underline{h} >0
		\end{cases}\label{sigma2}
	\end{equation}
	where the positive design parameters $\epsilon$ and $\bar t$ are free. The role of \eqref{sigma2} is to drive the state that is unsafe initially to the safe region by a prescribed time $\overline{D}+\bar{t}$.

Also, the selection of the design parameter $k_i$ is modified to account for model uncertainty by using the known parameter bounds from Assumptions \ref{assumboundd}--\ref{assum:theta_bounds}:
\begin{equation}
	k_i\ge\max_{\theta_1\in\bar{\Theta}_1,\theta_2\in\bar{\Theta}_2}\{0,\hat{k}_i(\chi(0);\theta_1, \theta_2)\},i=1,\cdots,n-1 \label{kcon}
\end{equation}
and $k_n >0$, where $\hat{k}_i$ is obtained from substituting the    variables $\theta_1,\theta_2$ for the unknown $\Theta$ in \eqref{hk},  which are specified by the bounds in Assumptions \ref{assum:theta_bounds} and  \eqref{thetabound}, respectively.  This choice guarantees the initial positivity of $h_i(z_i(D),D),i=2,\cdots,n$ from $h_1(z_1(D),D)>0$ ensured by \eqref{sigma2}.

To ensure robust safety against all possible parameter values before identification is complete, we construct the  following safe  controller under uncertainty:
\begin{equation}
	U_a(t)=\begin{cases}
		\frac{1}{\vartheta_0}\max\limits_{\theta_1\in\bar{\Theta}_1,\theta_2\in\bar{\Theta}_2}\vartheta_0\mathcal{U}(\chi(t);\theta_1,\theta_2),\quad 0  \le  t < \underline{D} \\
		\frac{1}{\vartheta_0}\quad\max\limits_{\theta_2\in\bar{\Theta}_2}\vartheta_0\mathcal{U}(\chi(t);\hat{\Theta}_1(\underline{D}),\theta_2),\quad   \underline{D}  \le  t  <  t_f,\\
		U_d(t). \phantom{aaaaaaaaaaa=======}  t\ge t_f
	\end{cases}\label{ua}
\end{equation}
where
\begin{equation}
	\vartheta_0:=\operatorname{sgn}(\vartheta(e(0),0)).\label{var0}
\end{equation}
 The switching time $t_f$ is the triggering time when the exact
parameter identification is achieved,  determined by \eqref{tf}. The diagram of the closed-loop system about the proposed data-driven safe control is depicted in Fig.\ref{fig:diag1}, and the result is presented in the next subsection.  \begin{figure}[t]
	\centering	
	\includegraphics[width=1\linewidth]{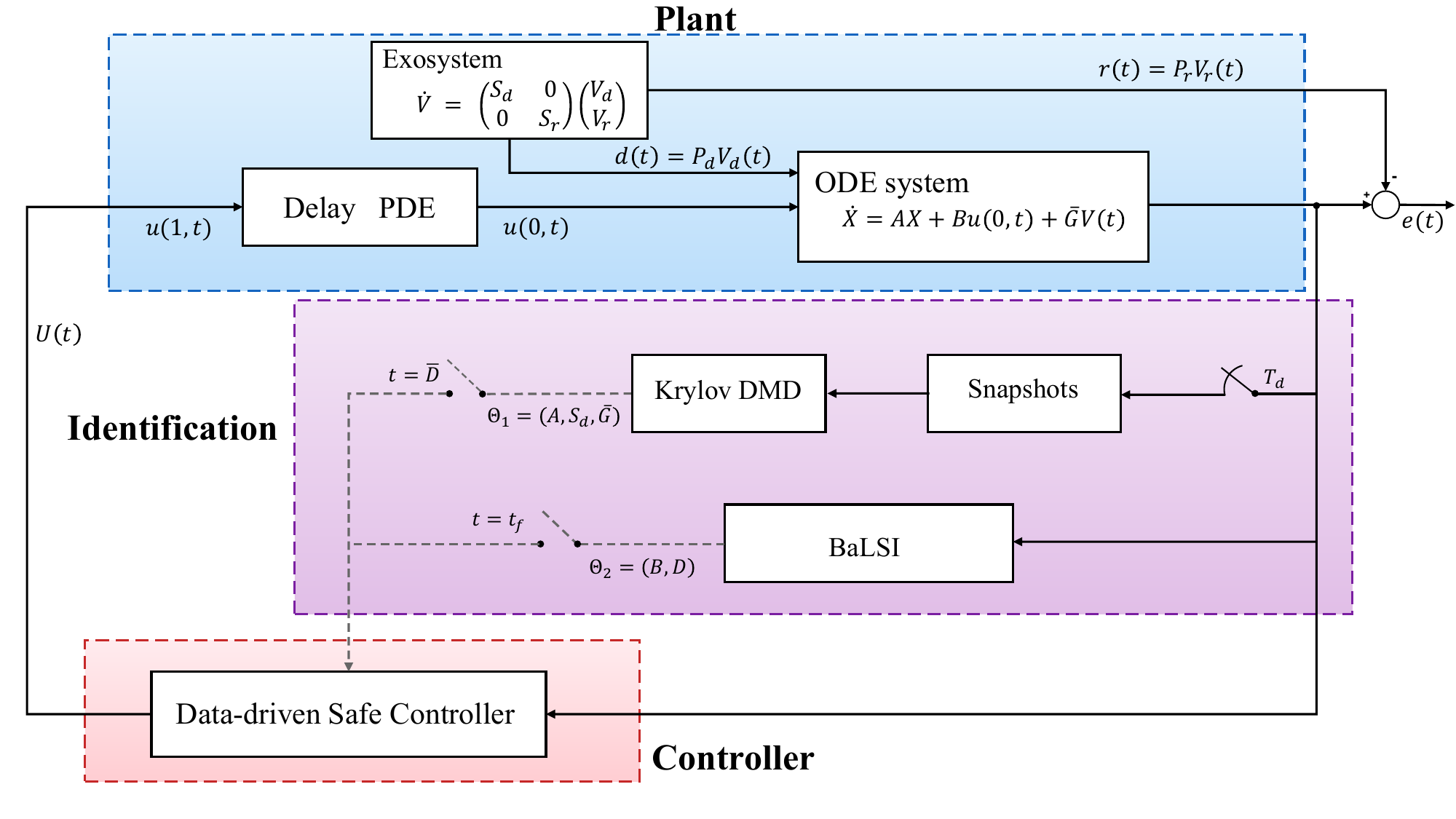}
	\caption{ Diagram of the proposed data-driven safe control system.}		
	\label{fig:diag1}
\end{figure}

\subsection{Result with the controller under uncertainty}
\begin{definition}\label{def1}
The initial safety condition is defined as: the safety is ensured in the uncontrolled period $t\in[0,D]$,  i.e., $h(e(t),t) \ge 0, \forall t \in [0,D]$ and $h(e(D),D) \neq 0$. It means that the safety holds for the uncontrolled period at the beginning and the error state $e(t)$ does not stay at the safe boundary at the initial regulation time $t=D$.
\end{definition}
\begin{theorem}\label{theo1}
	For initial data $X(0)\in\mathbb{R}^n$, $V(0)\in\mathbb{R}^{n_v}$, and the safety constraint $h(e,t)$ satisfying Assumptions \ref{assumh}, \ref{assumgrowth}, \ref{assumsign},  with the design parameters $k_1,\cdots,k_n$ satisfying \eqref{kcon},  the closed-loop system consisting of the disturbed linear plant \eqref{org1}, \eqref{org2} with unknown $\Theta$ \eqref{theta}, and the data-driven safe controller \eqref{ua},  has the following properties\par
	\noindent 1) There exists a unique weak solution $(X,V,u)\in C([0,\infty);\mathbb R^n\times\mathbb R^{n_v}\times L^2(0,1))$.\par
	\noindent 2)		The identification $\hat{\Theta}(t)$ of unknown parameters ${\Theta}$    is bounded and reaches the true value in finite time $t_f$ that is defined in \eqref{tf}.\par
	\noindent 	3) Tracking error state $e(t)$ \eqref{e} is convergent to zero and all states are bounded;\par
	\noindent 4) The  safety is ensured in the sense that
	
	a) If the initial safe condition in Definition \ref{def1} is met, then safety is guaranteed for $t \in [0, \infty)$.
	
	b) If the initial safe condition in Definition \ref{def1} is not met, but $h(e(D),D) >0$, then safety is guaranteed for $t \ge D$.
	
	c)  If the initial safe condition in Definition \ref{def1} is not met, and $h(e(D),D) \le0$, the safety is guaranteed for $t\ge\overline D+\bar{t}$ where $\bar{t}>0$ can be arbitrarily assigned by users.
	
\end{theorem} 
\begin{proof}
	1) 
	From the initial condition \eqref{u0} and the transformations \eqref{Z}, \eqref{h1}, \eqref{hi}, as well as  the exogenous signal model \eqref{modv}, we have $w[0]\in L^2(0,1)$, $H(0)\in\mathbb{R}^n$, and $V(t)\in\mathbb{R}^{n_v}$. 
	Based on Assumption \ref{assumh}, we know that the nominal control law \eqref{u} is locally Lipschitz in $X$, $V$, and $u[t]$. For the data-driven controller \eqref{ua} under uncertainty,  it preserves this property. By the method of characteristics and successive approximations, following the proof of [\citen{delay9}, Proposition 1], we obtain a unique weak solution on a sufficiently short time interval. Considering an arbitrary finite $T>0$,  it follows from
\eqref{pbound} in Assumption \ref{assumgrowth} and \eqref{u}   that  $|\mathcal U(\chi(t);\theta_1,\theta_2)|\le \xi_T(1+|X(t)|+|V(t)|+\|u[t]\|),0\le t< T$
for some constant $\xi_T>0$. The controller \eqref{ua} satisfies the same bound.
 Next, we define the norm as $\Omega(t)=|X(t)|^2+|V(t)|^2+\|u[t]\|^2$. Taking the derivative of $\Omega(t)$ along
\eqref{org11}--\eqref{org12} and using the preceding control bound,
we obtain the following estimate by integration by parts
and Young's inequality:
$
\dot{\Omega}(t)\le \xi_\Omega(1+\Omega(t)),
\quad 0\le t< T,
$
  for some $\xi_\Omega>0$.
It further implies that
$
\Omega(t)\le (\Omega(0)+1)e^{\xi_\Omega t}-1.
$
Hence, the system states remain bounded during parameter identification. Defining  the difference between $U_a$ \eqref{ua} and the nominal controller \eqref{u} as
	$
		\eta(t)=U_a(t)-U(t),$
	and applying the designed controller $U_a(t)$ to the plant \eqref{org11}--\eqref{org13},  we obtain the target system \eqref{tar1}, \eqref{tar2} with
	\begin{equation}
		w(1,t)=	\eta(t), \label{eta2}
	\end{equation}
where the difference  arises from the parameter mismatch error and  contains system signals $X(t)$, $V(t)$, and $u[t]$.
Applying the transformations \eqref{Z}, \eqref{h1}, \eqref{hi} along with the predictions \eqref{predx} and Assumption \ref{assumgrowth}, we obtain $\eta^2(t)\le\mathcal{N}_T(1+\Omega(t))$ for some constant $\mathcal{N}_T>0$. 
 Thus, $\eta(t)$ is bounded on $[0,t_f)$ and vanishes for $t\ge t_f$ due to the controller switching in \eqref{ua}. Consequently, we obtain $w[t]\in L^2(0,1)$ with $w[0]\in L^2(0,1)$ from \eqref{eta2}, \eqref{tar2}.

	Recalling the target system \eqref{tar1}, \eqref{tar2}, \eqref{eta2}, we obtain the solution of  the ODE as
	$
	H(t)=e^{A_ht}H(0)+\int_{0}^{t}e^{A_h(t-\tau)}B\vartheta(e(\tau),\tau)w(0,\tau)\diff \tau
	$. 
Together with the boundedness of $w(0,t)$ and $\vartheta(e(t),t)$ from Assumption \ref{assumh} on $[0,t_f+D]$, it follows from \eqref{tar1} that $H(t)$ is bounded on this interval. Since $w(0,t)=0$ for $t\ge t_f+D$ and $A_h$ is Hurwitz, we further obtain that $H(t)$ remains bounded for all $t\ge0$.
	Following the recursive process in the proof of Lemma \ref{lembound} shown in Appendix \ref{appbound}, we know that all $\frac{\partial h_i}{\partial z_j},1\leq j\leq i\leq n,i,j\in\mathbb{N}$ are bounded. We further conclude that $Z(t)\in \mathbb{R}^n$ considering $H(t)\in \mathbb{R}^n$ and \eqref{h1}--\eqref{sigma}. Then, $X(t)\in \mathbb{R}^n$ is obtained from \eqref{Z}.
	Applying the Cauchy-Schwarz inequality to the inverse transformation \eqref{inverw}, we have that $u[t]\in L^2(0,1)$ is then obtained from $w[t]\in L^2(0,1)$.  The states at each update serve as the initial data for the next interval. The  update period $t_{i+1}-t_i=T_a>0$ and the preceding bounds yield a unique weak solution $(X,V,u)\in C([0,\infty);\mathbb R^n\times\mathbb R^{n_v}\times L^2(0,1))$, as in [\citen{delay9}, Corollary 1]. 
{Recalling Assumption \ref{assumh} and the continuity of the states, the input \eqref{ua} is continuous except possibly at $t=\underline D$, $t=t_f$, and the parameter update times $t_i$, where it admits finite one-sided limits. Since  $T_a>0$, only finitely many such times occur on each finite interval. It follows that $U_a$ is piecewise continuous and Borel measurable on $[0,\infty)$.} The proof of this property is complete.\par
	2) According to Lemmas \ref{lemdmd}, \ref{lemada}, we directly obtain the property 2) that the	parameter identification $\hat{\Theta}(t)$ achieves the true value ${\Theta}$ in finite time $t_f$.\par
	3) In the target system \eqref{tar1}, \eqref{tar2} with \eqref{eta2}, the difference $\eta$ in \eqref{eta2} will vanish after a finite time $t=t_f$, recalling Property 2) and \eqref{eta2}. Together with  the fact  that the matrix $A_h$ in \eqref{doth} is Hurwitz, $|H(t)|$ is thus exponentially convergent to zero. 
Then,  following the same process in the proof of Lemma \ref{lembound} shown in Appendix \ref{appbound},   this property  is obtained.

	4) For the target system \eqref{tar1}, \eqref{tar2} with \eqref{eta2}, the sign of $\vartheta(t)$ defined in \eqref{var}  is constant, because Assumption \ref{assumh} implies that the function $\vartheta(e(t),t)$  is continuous and nonzero. As a result, $\operatorname{sgn}(\vartheta(t))$ is equal to $\vartheta_0$ defined in \eqref{var0}.  { For $\vartheta_0=-1$, the signed maximum in \eqref{ua} is $-\max(-\mathcal U)=\min\mathcal U$, while $\vartheta_0=1$ gives $\max\mathcal U$. Since the true parameters belong to the uncertainty set, both cases yield $
		\vartheta_0\eta(t)=\vartheta_0(U_a(t)-U(t))\ge0,\quad 0\le t<t_f.$ 
Recalling \eqref{tar2} and \eqref{eta2}, the characteristic formula yields $w(0,t)=w(1,t-D)=\eta(t-D)$ for $t\ge D$. Consequently, it follows that
$
		b\vartheta(t)w(0,t)=b|\vartheta(t)|\,\vartheta_0\eta(t-D)\ge0, t\in[D,t_f+D)$.
 For $t\ge t_f+D$, one has $w(0,t)=w(1,t-D)=0$ because the controller has switched to the identified nominal law and $\eta(t)=0$ for $t\ge t_f$.}
	Recalling \eqref{kcon},  the positivity of the initial barrier function value, i.e., $h_{i}\left(\underline{z}_{i}\left(D\right), D\right)$, is guaranteed by design parameters $k_i$  which belong to the subset of \eqref{k}. Therefore, from the structure of $H(t)$ in \eqref{tar1}, we know that all barrier functions $h_i$ are nonnegative during the time period $t\in[D,\infty)$, i.e., $h_1(t) \ge 0$  is ensured for all $t \ge D$. Next, we  prove property 4) about safety by examining each of the three cases:\par

	a) Because the initial safe condition in Definition \ref{def1} is met, \eqref{sigma2} directly  implies $\varsigma(t)\equiv0$. Consequently, by \eqref{h1}, $h(e(t),t)=h_1(z_1(t),t)\ge0$ for all $t\ge D$. Together with the initial condition, this establishes system safety, $h(e(t),t) \ge 0$, for all time $t \in [0,\infty)$.\par

	b)  Even though the safety is not guaranteed during the period $[0,D)$ in this case, considering $h(e(D),D) \ge 0$ and through the same process  as in the proof   of a), we have $h(e(t),t) \ge 0$  on $t \in  [D, \infty)$.
	
	c) We know that $\varsigma(t)$ is identically zero after a preset finite time $t \ge\overline{D}+\bar{t}$ according to \eqref{sigma2}. It follows  from \eqref{h1} that $h(e(t),t) = h_1(z_1(t),t)\ge0$ for all $t \ge \overline D+\bar{t}$.  Thus, safety is guaranteed after a finite time  $t \ge \overline D+\bar{t}$. 
	
	The proof of this theorem is complete.
\end{proof}
\section{Control Design using output data}\label{secout}
In contrast to the preceding full-state feedback design, this section addresses a more practical situation in which only the output state is measurable. We develop an output-feedback safe controller by reconstructing the system and disturbance dynamics directly from output measurements.

Given the very limited amount of data available for identification, in addition to Assumptions \ref{assumboundd}, \ref{assum:theta_bounds},  it is necessary to impose the following additional assumptions. For the original plant \eqref{org1}, \eqref{org2}, and the exogenous model \eqref{modd}--\eqref{modv}, we now assume that the coefficient matrices $B, G, P_d$ are known. The system matrices $A, S_d$ of the plant and disturbances remain unknown, but are assumed to take the following companion form:
\begin{equation}
	A\! =\!\! \begin{pmatrix}
  0 & 1 & & \\
    & \ddots & \ddots & \\
    & & 0 & 1 \\
  a_1 & \dots & a_{n-1} & a_n
\end{pmatrix}\!,\!
S_d\!=\!\! \begin{pmatrix}
  0 & 1 & & \\
    & \ddots & \ddots & \\
    & & 0 & 1 \\
  \beta_1 & \dots & \beta_{n_d-1} & \beta_{n_d}
\end{pmatrix}\!.\label{a2}
\end{equation}
Accordingly, the unknown argument $\Theta$ \eqref{theta} becomes
\begin{equation}
	\Theta=\{\underbrace{A,S_d}_{\Theta_1},D\}.\label{theta2}
\end{equation}
Additionally, the following assumption ensures the unknown initial states are within a known compact set $\bar{\zeta}_0$, which is required for ensuring safety before the plant states and disturbances are successfully estimated by the data-driven observer.
\begin{assumption}\label{assumboundv}
	The bounds of the initial states $X(0)$ of the plant and the initial values $V_d(0)$ of the exogenous signal are known and arbitrary.
\end{assumption}
Write $\hat V=\operatorname{col}(\hat V_d,V_r)$ and $\tilde V=V-\hat V=\operatorname{col}(\tilde V_d,0)$, where $\hat V_d$ is the disturbance estimate constructed below. The known reference state is not estimated.

The safe barrier function $h$  satisfies Assumptions \ref{assumh}, \ref{assumsign}.
Besides, given that the system state and disturbances are unmeasured in this section, we impose an additional assumption for the barrier function \mbox{$h$}:
\begin{assumption}\label{assumh3}
The barrier function $h$ ensures that there exists a positive $\xi_e$ such that
\begin{equation}
	|\tilde{\mathcal{P}}(\tilde{V}(t),\tilde{X}(t))|\leq\xi_{e}(|\tilde{V}(t)|+|\tilde{X}(t)|)\label{tilp}
\end{equation}
where $\tilde{\mathcal{P}}:= \bar{\mathcal{P}}(T_{z}X(t)+T_{v}V(t), u[t],V(t),t,D)-\bar{\mathcal{P}}(T_{z}\hat{X}(t)+T_{v}\hat{V}(t), u[t],\hat{V}(t),t,D)$
and
$\bar{\mathcal{P}}$ \eqref{p} is determined by  \eqref{h1}--\eqref{var}.
\end{assumption}

\subsection{Identification of the unknown system dynamics}
In this section, we identify the system parameters from input-output data to construct an observer for the unmeasured states and disturbances.
During $t\in[0,\underline{D}]$, we use the extended system \eqref{tilA} and the Koopman elements  from Sec. \ref{secdmd}. We collect output snapshots from $Y(t)$ in \eqref{org2} at the sampling period $T_d$ satisfying Assumption \ref{assumtd} and Remark \ref{remsam} to form
\begin{equation}
\tilde{\mathcal{Y}}=
\begin{pmatrix}
	Y(0)\quad Y(1)\quad \cdots\quad Y(2\tilde{n}-1) \\
\end{pmatrix}\in\mathbb{R}^{1\times2\tilde{n}}\label{snap2}
\end{equation}
with snapshots $Y(k) =  Y(kT_d), k\in\mathbb{N}$.

Accordingly, the output is expressed using the extended state $\tilde{X}$  as 
$	Y={\bar{C}}\tilde{X}$, where $\bar{C}={\begin{pmatrix}0&C\end{pmatrix}}$. 
From the output data sequence in \eqref{snap2}, we form the Hankel matrix
\begin{equation}
H_{\tilde{n}}( Y)=\begin{pmatrix} Y(0)& Y(1)&\ldots& Y(\tilde{n}-1)\\ Y(1)& Y(2)&\ldots& Y(\tilde{n})\\\vdots&\vdots&\vdots&\vdots\\ Y(\tilde{n}-1)& Y(\tilde{n})&\ldots& Y(2\tilde{n}-2)\end{pmatrix}\in\mathbb{R}^{\tilde{n}\times\tilde{n}}.\label{hankely}
\end{equation}
This Hankel matrix can also be factorized into the product of the extended observability and controllability matrices, i.e., $H_{\tilde{n}}(Y)=Q_oQ_c$, where $Q_c$ is given by \eqref{qc} and $Q_o$ is the extended observability matrix
\begin{equation}
    Q_o = \begin{pmatrix} \bar{C}^\top & (\bar{C}\tilde{A}_d)^\top & \dots & (\bar{C}\tilde{A}_d^{\tilde{n}-1})^\top \end{pmatrix}^\top\in\mathbb{R}^{\tilde{n}\times\tilde{n}}.\label{qo2}
\end{equation}

{Unlike \cite{data2}, which requires the output dimension to be no less than the disturbance dimension, the following lemma establishes a full-rank condition for $H_{\tilde{n}}(Y)$ using only low-dimensional output data.}

\begin{lemma}\label{lemfull2}
If
\begin{equation}
	\operatorname{rank}\begin{pmatrix}A-\tilde{\lambda} I&GP_d\\C&0\end{pmatrix}=n+1,\quad\forall\tilde{\lambda}\in\operatorname{eig}(S_d),\label{rank1}
\end{equation}
the Hankel matrix $	H_{\tilde{n}}(Y) $
satisfies
\begin{equation}
	\operatorname{rank}	H_{\tilde{n}}({Y}) =\tilde{n}\label{ranhy}
\end{equation}
for almost all $Y(0)\in\mathbb{R}$.
\end{lemma}

\begin{proof}
	The proof is shown in Appendix \ref{appfull2}.
\end{proof}

\begin{lemma}
If $\operatorname{rank}H_{\tilde{n}}(Y)=\tilde{n}$,
\begin{equation}
	\mathcal{K}(T_d)H_{\tilde{n}}(Y)=H_{\tilde{n}}(Y)F
\end{equation}
holds, where $F$ has the companion form in \eqref{kh},
with $
f=\mathrm{col}(f_{0},\ldots,f_{\tilde{n}-1})=-(H_{\tilde{n}}^{\top}(Y)H_{\tilde{n}}(Y))^{-1}H_{\tilde{n}}^{\top}(Y)\allowbreak\mathrm{col}(Y(\tilde{n}),\ldots,Y(2\tilde{n}-1))$.
\end{lemma}

\begin{proof}
Its proof can be found in [\citen{data2}, Lemma 3].
\end{proof}

Next, we extract the Koopman eigenvalues from matrix $F$, and subsequently identify the system dynamic matrix under the following assumption.

\begin{assumption}\label{assumspecsep}
{The disturbance dimension $n_d$ is known, and $\operatorname{eig}(A)\cap(\mathrm i\mathbb R\setminus\{0\})=\emptyset$.}
\end{assumption}

\begin{lemma}\label{lemdmd2}
Assume that $\operatorname{rank}H_{\tilde{n}}(Y)=\tilde{n}$ {and Assumption \ref{assumspecsep} holds}, and let $\hat{\nu}_i,i=1,\cdots,\tilde{n}$ be the eigenvectors of $F$  corresponding to the eigenvalues $\hat{\lambda}_i$.
The Koopman eigenvalues can be computed from
\begin{equation}
	\tilde{\lambda}_i=\frac{\ln\hat{\lambda}_i}{T_d},\quad i=1,\ldots,\tilde{n}. \label{tillam}
\end{equation}
 {The identified eigenvalues are classified and relabeled such that
$\tilde{\lambda}_i$, $i=1,\ldots,n_d$, belong to $S_d$,
and $\tilde{\lambda}_i$, $i=n_d+1,\ldots,\tilde{n}$, belong to $A$.}
The exact identification  of the unknown system dynamics ${\Theta}_1=\{A, S_d\}$ is achieved at $t=\underline{D}$, that is $\hat{\Theta}_1(\underline{D})=\{\mathcal{A},\mathcal{S}_d\}\equiv{\Theta}_1$, where
\begin{equation}
       \mathcal{A} = V_A \Lambda_A V_A^{-1},\qquad    \mathcal{S}_d= V_{d} \Lambda_{d} V_{d}^{-1},  \label{mathsd2}
\end{equation}
    with $\Lambda_A = \operatorname{diag}(\tilde{\lambda}_{n_d+1}, \cdots, \tilde{\lambda}_{\tilde{n}})$, $\Lambda_{d} = \operatorname{diag}(\tilde{\lambda}_1, \cdots, \tilde{\lambda}_{n_d})$, and the eigenvector matrices 
		\begin{equation}
    V_A = \begin{pmatrix}
        1 & \dots & 1 \\
        \tilde{\lambda}_{n_d+1} & \dots & \tilde{\lambda}_{\tilde{n}} \\
        \vdots & \ddots & \vdots \\
        \tilde{\lambda}_{n_d+1}^{n-1} & \dots & \tilde{\lambda}_{\tilde{n}}^{n-1}
    \end{pmatrix},
    V_d = \begin{pmatrix}
        1 & \dots & 1 \\
        \tilde{\lambda}_1 & \dots & \tilde{\lambda}_{n_d} \\
        \vdots & \ddots & \vdots \\
        \tilde{\lambda}_1^{n_d-1} & \dots & \tilde{\lambda}_{n_d}^{n_d-1}
    \end{pmatrix},
    \label{eq:vandermonde_matrices}
\end{equation}
	as Vandermonde matrices constructed from the corresponding eigenvalues.	
\end{lemma}
\begin{proof}
The proof is shown in Appendix \ref{appdmd2}.

\end{proof}

In addition to \eqref{mathsd2}, alternatively, the coefficients $a_i, \beta_i$ in the matrices $A$, $S_d$ \eqref{a2} can be obtained by the coefficients of the characteristic polynomial using all known eigenvalues $\tilde{\lambda}_i$.

To estimate the unknown delay parameter $D$, we adapt the same BaLSI design in Sec. \ref{secada} to build the delay identifier.
Since only $D$ is unknown, the general parameter estimator \eqref{est1} simplifies to
\begin{align}
&\hat{D}(t_{i+1})=\arg\min\Big\{|\ell-\hat{D}(t_{i})|^2: \ell\in{D}_0,\notag\\
&{F}_n(t_{i+1},\mu_{i+1})={Q}_n(t_{i+1},\mu_{i+1})\ell,n=1,2,\cdots\Big\},\label{est2}
\end{align}
where ${F}_n$, ${Q}_n$ are defined as in  \eqref{gn}, and the set $
{D}_0:=\{\ell\in \mathbb{R}:\underline{D}\leq \ell\leq\overline{D}\} 
$
is given in Assumption \ref{assumboundd}.
\begin{lemma}\label{lemada2}
The finite-time exact identification of the unknown
 parameter $D$ is achieved, that is, there exists a finite time $t_f$
defined by 
\begin{equation}
	t_f=\min\{t_i:\exists t\in [0,t_i),u[t]\neq0\}, \label{tf2}
\end{equation}
such that $\hat{D}(t_f)=\mathcal{D}\equiv D, \forall t\ge t_f$.
\end{lemma}
\begin{proof}
{The proof is shown in Appendix \ref{appbalsi2}.}
\end{proof}

\subsection{{Design of observer under uncertain dynamics}}
We now construct an observer for $X(t)$ and $V_d(t)$ using the identified parameters and the measured output $Y(t)$.
Treating the disturbances as extended states, the state observer for the extended system consisting of \eqref{org1}, \eqref{org2}  \eqref{modd} is designed as
\begin{align}
\dot{\hat{X}}(t)=&\hat{A}(t)\hat{X}(t)+BU(t-\hat{\mathcal{D}}(t))+{G}P_d\hat{V}_d(t)\notag\\
&+L_x(t)(Y(t)-\hat{Y}(t)),\label{ox}\\
\dot{\hat{V}}_d(t)=&\hat{S}_d(t) \hat{V}_d(t)+L_v(t)(Y(t)-\hat{Y}(t)),\label{ov}\\
\hat{Y}(t)=&C\hat{X}(t),\label{oy}
\end{align}
where  $\hat{X}$ is the estimate of the state vector $X$, $\hat{V}_d(t)$ is the estimate of  the unmeasured disturbance $V_d(t)$, and the observer gains  $L_v(t)$ and $L_x(t)$ are designed later.
The system matrices used in the observer are the identified values of $A$ and $S_d$, denoted as $\hat{A}(t)$ and $\hat{S}_d(t)$. Specifically, as established in Lemma \ref{lemdmd2}, these estimates are piecewise constant and converge in finite time: $\hat{A}(t)=\begin{cases}
A_0,& 0\le t< \underline{D}\\
\mathcal{A},& t\ge \underline{D}\\
\end{cases},\hat{S}_d(t)=\begin{cases}
{S_d}_0,& 0\le t< \underline{D}\\
\mathcal{S}_d,& t\ge \underline{D}\\
\end{cases},$ and $\hat{\mathcal{D}}(t)=\begin{cases}
\mathcal{D}_0,& 0\le t<t_f\\
\mathcal{D},& t\ge t_f
\end{cases}$, where $\mathcal{A},\mathcal{S}_d$ are given in Lemma \ref{lemdmd2}, $\mathcal{D}=\hat{D}(t_f)$ is given in Lemma \ref{lemada2}, and $A_0,{S_d}_0,{\mathcal{D}}_0$ are arbitrary initial estimates satisfying Assumptions \ref{assumboundd}, \ref{assum:theta_bounds}.

Defining the state estimation error as
\begin{equation}
{e_o}^\top:=(\tilde{V}_d, \tilde{X})=({V}_d,{X})-(\hat{V}_d,\hat{X}),\label{eo}
\end{equation}
the dynamics of the extended observer error are given by
\begin{align}
\dot{e_o}
=\underbrace{\begin{pmatrix}
		S_d&-L_v(t)C\\
		GP_d&A-L_x(t)C
\end{pmatrix}}_{L(t)}
e_o
+\underbrace{\begin{pmatrix}
		\Delta{S}_d(t)\,\hat{V}_d(t)\\
		\Delta A(t)\,\hat{X}(t)+B \Delta U(t)
\end{pmatrix}}_{\Delta{e}(t)},\label{oerr}
\end{align}
where $\Delta{e}(t)$ arises entirely from the parameter identification errors for the system matrices and the time delay,
where $\Delta{S}_d(t):=S_d-\hat{S}_d(t)$, $\Delta A(t):=A-\hat{A}(t)$, $\Delta U(t) :=U(t-D)-U(t-\hat{\mathcal D}(t))$.
The observer gains, $L_v(t)$ and $L_x(t)$, are piecewise constant  and are defined as 
\begin{equation}
L_v(t)=\begin{cases}
	L_{v,0},& 0\le t< \underline{D}\\
	\mathcal{L}_v,& t\ge \underline{D}\\
\end{cases},
L_x(t)=\begin{cases}
	L_{x,0},& 0\le t< \underline{D}\\
	\mathcal{L}_x,& t\ge \underline{D}\\
\end{cases}\label{lgain}
\end{equation}
where the initial observer gains  $L_{v,0}$, $L_{x,0}$ are free design parameters, and $\mathcal{L}_v$, $\mathcal{L}_x$ are determined by using the estimates ${\mathcal{A}, \mathcal{S}_d}$ to make the matrix $L(t)$ Hurwitz, denoted as $\mathcal{L}$. {It indicates that the whole observer gains $L(t)$ are piecewise-constant as well, i.e., $L(t)=\begin{cases}
	L_{0},& 0\le t< \underline{D}\\
	\mathcal{L},& t\ge \underline{D}\\
\end{cases}.$}

We present the following lemmas to estimate the upper bound of the observation error, which is required for the subsequent safe output-feedback design.
\begin{lemma}\label{lemerr1}
For a square matrix $L$, we have the following estimate
$
	\Vert e^{Lt} \Vert\le M_L(t)e^{-\delta t}, \forall t\ge 0 \label{eqlemerr1}
$
where $\delta=-\max \Re(\operatorname{eig}(L))$, and where $M_L =\|P_1\|\|P_1^{-1}\|>0$, which is constant when $L$ is diagonalizable.  Otherwise, $M_L(t)=\|P_{1}\|\|P_{1}^{-1}\|\sum_{j=0}^{n_{k}}\frac{t^{j}}{j!}>0$  is a polynomial in $t$. The matrix $P_1$ is invertible, satisfying $L=P_1\Lambda{P_1}^{-1}$ where $\Lambda$ is either a diagonal matrix (if $L$ is diagonalizable)
or a Jordan matrix (in which the size of the largest Jordan block is $n_k$).
\end{lemma}
\begin{proof}
Its proof can be found in [\citen{outsafe}, Appendix 0.4].
\end{proof}

Defining the observer error norm
\begin{equation}
\Omega_e=|\tilde{X}|+|\tilde{V}_d|,\label{omega}
\end{equation}
we establish the next lemma.

\begin{lemma}\label{lemerr}
For initial data $\hat{X}(0)\in\mathbb{R}^n$, $\hat{V}_d(0)\in\mathbb{R}^{n_d}$,   and an observer \eqref{ox}--\eqref{oy} where the observer gain  functions $L_v(t)$, $L_x(t)$	 are defined by \eqref{lgain} and satisfies   that $L(t)\equiv \mathcal{L}, t\ge\underline{D}$ is Hurwitz,  the observer error system \eqref{oerr} is well-posed in the sense of $\tilde{X}(0)\in\mathbb{R}^n$, $\tilde{V}_d(0)\in\mathbb{R}^{n_d}$, for $t\in[0,\infty)$. Moreover, the exponential convergence to zero of the observer errors \eqref{oerr} is achieved in the sense that
\begin{align}
	\Omega_e&\le \sqrt2\Gamma_{e,1}(t):=\rho_{e_0}(t),\quad 0\le t< \underline{D};\label{omee}\\
	\Omega_e&\le \sqrt2\Gamma_{e,2}(t):=\rho_{e_1}(t),\quad \underline{D}\le t< t_f;\label{omee1}\\
	\Omega_e&\le \sqrt2\Gamma_{e,2}(t_f) M_{\mathcal{L}}(t-t_f)e^{-\delta_{\mathcal{L}}(t-t_f)}:=\rho_e(t),\quad t\ge t_f\label{omee2}
\end{align}
where the observer error bound functions $\Gamma_{e,1}(t)$, $\Gamma_{e,2}(t)$ are given in Appendix \ref{appobserver}, and $M_{\mathcal{L}},\delta_{\mathcal{L}}$ are 
obtained from applying Lemma \ref{lemerr1} to the Hurwitz matrix $\mathcal{L}$.
\end{lemma}

\begin{proof}
The proof is shown in Appendix \ref{appobserver}.
\end{proof}

\subsection{Safe output regulation under unknown parameters}\label{secconobe}

With the modified dynamic matrices $A, S_d$ in \eqref{a2}, the first  transformation \eqref{Z} in Sec. \ref{secfistran} reduces to
$
Z(t)=X(t)+T_vV(t)
$. The subsequent transformation design then proceeds as described in Sec. \ref{secnom}. 
With the estimates $\hat{X}(t)$ and $\hat{V}(t)$ from the observer \eqref{ox}--\eqref{oy}, we can construct the output-feedback controller:
\begin{equation}
\hat{U}_d(t) := \mathcal{U}(\hat{\chi}(t); \hat{\Theta}(t))+\vartheta_0 \hat\rho_e(t;\hat{\Theta}(t)), \label{uod}
\end{equation}
where  $\mathcal{U}(\hat\chi(t);\Theta)$ is  derived by
employing the observer state $\hat{\chi}(t) = ( \hat{V}_d(t),V_r(t), \hat{X}(t))$ to replace the system state \eqref{chi}
in \eqref{u}. 
Following the approach in \cite{outsafe}, the following function $\hat\rho_e(t)$ with the known constant $\vartheta_0$ given by \eqref{var0} is added in \eqref{uod} to tolerate the observer error \eqref{eo},
\begin{equation}
\hat\rho_e(t)=\begin{cases}
	0, 0\le t< t_f\\
	2\max\{\|\tilde{K}e^\mathcal{D\mathcal{A}}\|,|\bar{\gamma}(1,\hat{\Theta}(t_f))|,\xi_e\}\rho_{e}(t), t\ge t_f
\end{cases}\label{rho2}
\end{equation} where the gains $\tilde{K}$ and $\bar{\gamma}(1,\hat{\Theta}(t_f))$ are obtained by substituting $\hat{\Theta}(t_f)$ into \eqref{K} and \eqref{bargam}, $\rho_e(t)$ is given in \eqref{omee2},  and where $\xi_e$ is a constant defined in Assumption \ref{assumh3}.

 Defining the  vector $\zeta_0$ that includes the initial states $X(0)$ and $V_d(0)$ , it belongs  to the set $\bar{\zeta}_0$ according to Assumption \ref{assumboundv}. The  parameters ${\theta}_1$, $\mathfrak{D}$  represent the unknown system parameters, bounded by $\bar{\Theta}_1,D_0$ from Assumptions \ref{assum:theta_bounds}, \ref{assumboundd}. On the basis of it, we substitute
the lower bound
$
\underline{h}=\min_{\zeta_0\in\bar{\zeta}_0,{\theta}_1\in\bar{\Theta}_1,\mathfrak{D}\in D_0}\mathcal{P}_0(\zeta_0;{\theta}_1,\mathfrak{D}) \label{uh2}
$
into \eqref{sigma2} for the uncertain value $h(e(D),D)$, to choose the positive constant $\epsilon$ in \eqref{sigma2} to ensure that $h_1(z_1(D),D)\ge0$. 

Analogously, the design parameters $k_i$ are selected only using the boundaries to robustly guarantee the initial positivity of all functions $h_i(z_i(D),D)$ across the entire uncertainty set
\begin{equation}
k_i\ge\max_{\zeta_0\in\bar{\zeta}_0,{\theta}_1\in\bar{\Theta}_1,\mathfrak{D}\in D_0}\{0,\hat{k}_i(\zeta_0;{\theta}_1,\mathfrak{D})\},i=1,\cdots,n-1,   \label{kcon3}
\end{equation}
$k_n >0$, where $\hat{k}_i$ is given by \eqref{hk}.

Following the same design principles as between \eqref{kcon}--\eqref{var0}, we formulate the following safe controller to operate under both parameter and state uncertainty:
\begin{equation}
U_a(t)=\begin{cases}
	\frac{1}{\vartheta_0}\max\limits_{\zeta(t)\in\bar{\zeta}(t),{\theta}_1\in\bar{\Theta}_1,\mathfrak{D}\in D_0}\vartheta_0\mathcal{U}(\zeta(t);{\theta}_1,\mathfrak{D}),& \\
	\phantom{=================}   0  \le  t < \underline{D} \\
	\frac{1}{\vartheta_0}\max\limits_{\zeta(t)\in\bar{\zeta}(t),\mathfrak{D}\in D_0}\vartheta_0\mathcal{U}(\zeta(t);\hat\Theta_1(\underline{D}),\mathfrak{D}),& \\
	\phantom{=================}   \underline{D}\le t    <  t_f,\\
	\hat{U}_d(t),\phantom{aaaaaaaaaaa=====}  t\ge t_f
\end{cases}\label{ua2}
\end{equation}
where the form of $\mathcal{U} $ is specified by \eqref{u}, and $\bar\zeta(t)$ denotes all possible values of  states, with the known reference component fixed at $V_r(t)$, obtained from the observer state $\hat\chi(t)$ with its upper bound defined in \eqref{omee}, \eqref{omee1}, i.e.,
\begin{equation}
\bar{\zeta}(t)=\begin{cases}
	\{\zeta(t)\in\mathbb{R}^{\tilde{n}}:|\zeta(t)-\hat\chi(t)|\le \rho_{e_0}(t)\}, & 0 \le t < \underline{D} \\
	\{\zeta(t)\in\mathbb{R}^{\tilde{n}}:|\zeta(t)-\hat\chi(t)|\le \rho_{e_1}(t)\}. & \underline{D} \le  t  <  t_f
\end{cases}
\label{barzeta2}
\end{equation}

\begin{figure}[t]
	\centering	
	\includegraphics[width=1\linewidth]{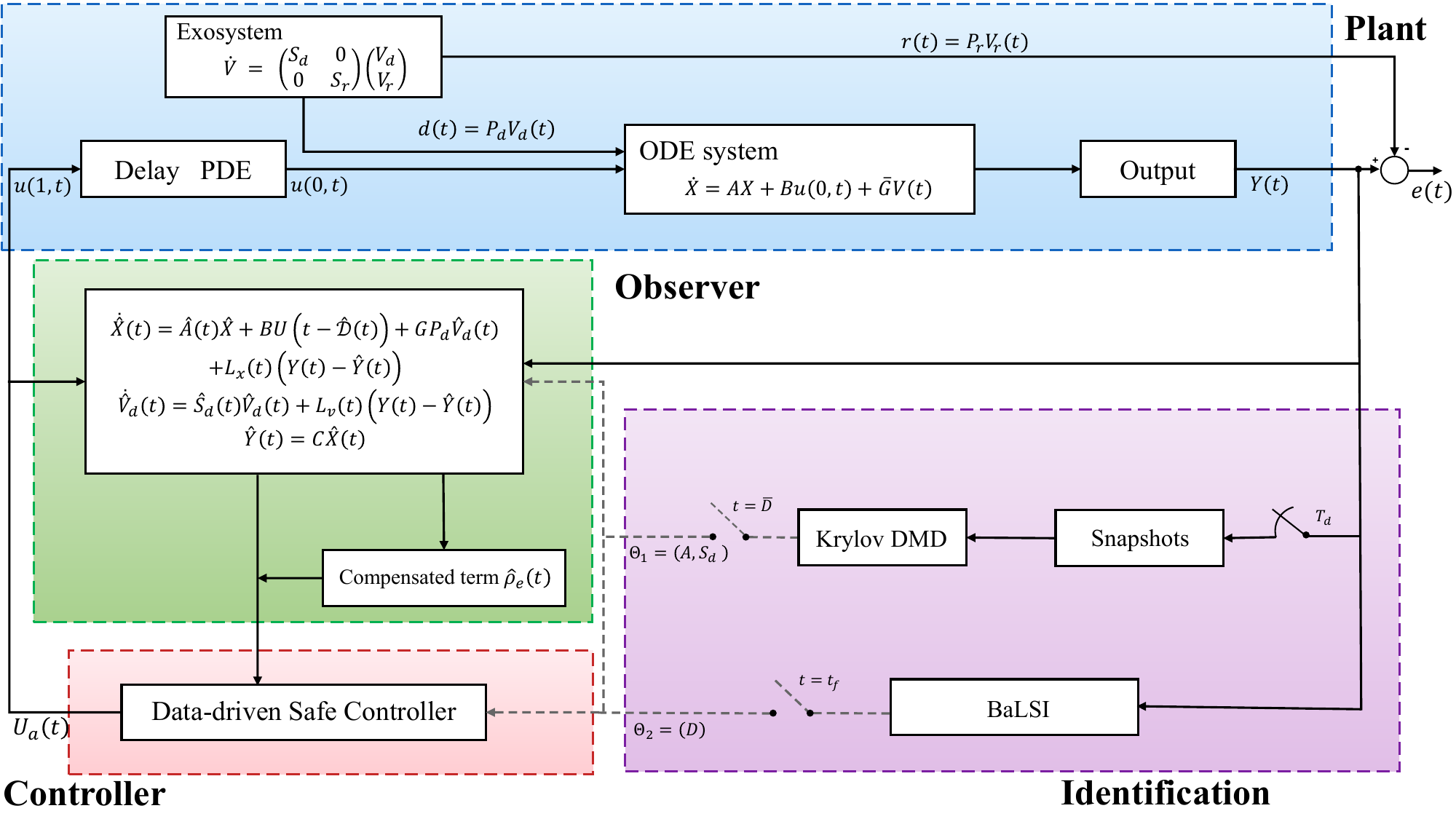}
	\caption{ Diagram of the proposed output-feedback data-driven safe control system.}		
	\label{fig:diag2}
\end{figure}

\begin{remark}
To strictly guarantee system safety over the time interval $[0, t_f)$ under unknown states and parameters, all possible values are computed. This process is potentially computationally demanding. However, if a control input that guarantees safety over the finite horizon $[0,\underline D)$ is 
known, the computational burden can be significantly reduced by simply applying this safe input during this initial interval. Such an input for a short initial time is always available in many practical applications.
\end{remark}

The diagram of the proposed safe output regulation closed-loop system is depicted in Fig. \ref{fig:diag2}, and the result is presented in the next subsection. 
\subsection{Result of safe output-feedback Controller}
\begin{theorem}\label{theo2}
For initial data $X(0)\in\mathbb{R}^n$,  $V(0)\in\mathbb{R}^{n_v}$, $\hat{X}(0)\in\mathbb{R}^n$, $\hat{V}_d(0)\in\mathbb{R}^{n_d}$  satisfying Assumption \ref{assumboundv}, and the safety constraint $h(e,t)$ satisfying Assumptions \ref{assumh}, \ref{assumgrowth}, \ref{assumsign}, \ref{assumh3},  choosing  the design parameters $k_1,\cdots,k_n$ satisfying \eqref{kcon3}, the closed-loop system consisting of the disturbed linear plant \eqref{org1}, \eqref{org2} with unknown $\Theta$ \eqref{theta2}, the observer \eqref{ox}--\eqref{oy}, and the controller \eqref{ua2}, has the following properties:\par
\noindent	 1) There exists a unique weak solution  $(X,V,u,\hat X,\hat V_d)\in C([0,\infty); \mathbb R^n\times\mathbb R^{n_v}\times L^2(0,1) \times\mathbb R^n\times\mathbb R^{n_d})$.\par
\noindent 2)		The identification $\hat{\Theta}(t)$ of unknown parameters ${\Theta}$    is bounded and reaches the true value in  finite time $t_f$ that is defined in \eqref{tf2}.\par
\noindent 	3) Tracking error state $e(t)$ \eqref{e} is convergent to zero and all states are bounded;\par
\noindent 4) Safety is ensured in the sense that

a) If the initial safe condition in Definition \ref{def1} is met, then safety is guaranteed for $t \in [0, \infty)$.

b) If the initial safe condition in Definition \ref{def1} is not met, but $h(e(D),D) >0$, then safety is guaranteed for $t \ge D$.

c)  If the initial safe condition in Definition \ref{def1} is not met, and $h(e(D),D) \le 0$, the safety is guaranteed for  $t\ge\overline D+\bar{t}$ where $\bar{t}>0$ can be arbitrarily assigned by users.

\end{theorem}
\begin{proof}
1)
From the initial condition \eqref{u0} and the transformations \eqref{Z}, \eqref{h1}, \eqref{hi}, as well as  the exogenous signal
model \eqref{modv}, we have $w[0]\in L^2(0,1)$, $H(0)\in\mathbb{R}^n$, and $V(t)\in\mathbb{R}^{n_v}$.
Under Assumption \ref{assumh}, the maximization over $\bar\zeta(t)$ in \eqref{ua2} preserves local Lipschitz continuity in $\hat\chi(t)$, $u[t]$. Following the proof of Property 1 in Theorem \ref{theo1}, we obtain a unique weak solution on a sufficiently short time interval. Recalling \eqref{ox}--\eqref{oy} and the proof of Lemma~\ref{lemerr}, it follows that the observer states and the bounds \eqref{omee}, \eqref{omee1} are continuous    between consecutive parameter update times  and have finite one-sided limits at the update times.
 The bound on the error $\Delta U(t)$  in the proof of Lemma~\ref{lemerr} depends only on the past input $U_a(s)$ for $s\in[-\overline D,t-\underline D]$.  Considering the linear dynamics of the plant \eqref{org1},  \eqref{modv} and observer \eqref{ox}--\eqref{oy}, together with the error bounds \eqref{omee}, \eqref{omee1},  we obtain that there exists a constant  $\xi_T>0$ such that $|X(t)|^2+|V(t)|^2+\|u[t]\|^2+|\hat X(t)|^2+|\hat V_d(t)|^2\le\xi_TQ(t)$ over each compact interval $[0,T]$ of existence before $t_f$, and the same bound holds for  the squared norms of the states in the uncertainty set \eqref{barzeta2}, where $Q(t)=1+\int_0^t\sup_{0\le s\le\tau}|U_a(s)|^2\diff\tau$.   In particular, we have $Q'(t)\le\xi_TQ(t)$ almost everywhere for some constant $\xi_T$ from \eqref{pbound}, \eqref{barzeta2}
 and the monotonicity of $Q$. It further implies that $Q(t)\le e^{\xi_Tt}$.
On the  interval $t\in[0,t_f)$,  we define the difference between $U_a(t)$ \eqref{ua2} and $U(t)$ \eqref{u} as
$\hat{\eta}_1(t)=U_a(t)-U(t),\label{hateta1}
$
which arises from both parameter mismatch and state estimation error. The preceding estimates and \eqref{u} imply that $\hat\eta_1(t)$ is bounded on $[0,t_f)$.  Note that, for the time period \mbox{$t\in[0,t_f)$}, the target system becomes \eqref{tar1}, \eqref{tar2} with \mbox{$w(1,t)=\hat{\eta}_1(t)$}. Thus, we have that \mbox{$w(1,t)^2$} is bounded during \mbox{$t\in[0,t_f)$}.
According to Lemmas \ref{lemdmd2}, \ref{lemada2}, all identified parameters have converged to their true values by the time $t=t_f$. Consequently, when the controller \eqref{ua2} switches at $t=t_f$ with the exact identification $\hat{\Theta}(t_f)$, the difference between the applied controller $U_a(t)$ \eqref{ua2} and the nominal controller $U(t)$ \eqref{u} is given by
$
	U_a(t)-U=\vartheta_0 \hat\rho_e(t;\hat{\Theta}(t_f))+\hat{\eta}_2(t),\label{hateta2}
$
where $\hat{\eta}_2(t)$ is defined as
$
	\hat{\eta}_2(t)=\tilde{K}e^{\mathcal{D}\mathcal{A}}\tilde X(t)-\bar{\gamma}(1;\hat{\Theta}(t_f))\tilde V(t)-\tilde{\mathcal{P}}(\tilde V(t),\tilde X(t)).\label{etao}
$
Accordingly, the target system becomes \eqref{tar1}, \eqref{tar2} with
\begin{equation}
	w(1,t)=\vartheta_0 \hat\rho_e(t)+\hat{\eta}_2(t),\quad t \ge t_f.\label{w1t}
\end{equation}
Recalling Lemma \ref{lemerr}, \eqref{p}, \eqref{tilp} in Assumption \ref{assumh3}, \eqref{rho2}, we know that $w(1,t)$ is bounded for all $t\in[t_f,\infty)$ from \eqref{w1t}. As it is also bounded on $[0,t_f)$, $w(1,t)$ is bounded for $t \ge 0$, which ensures $w[t] \in L^2(0,1)$. Following the same process in the proof of property 1) in Theorem \ref{theo1},  one obtains  $H(t)\in \mathbb{R}^n$, $Z(t)\in \mathbb{R}^n$, $X(t)\in \mathbb{R}^n$ and $u[t]\in L^2(0,1)$ as well.
Recalling the results about the boundedness of observer errors in Lemma \ref{lemerr}, and \eqref{eo}, we further conclude that the closed-loop system  admits a unique weak solution  $(X,V,u,\hat X,\hat V_d)\in C([0,\infty); \mathbb R^n\times\mathbb R^{n_v}\times L^2(0,1) \times\mathbb R^n\times\mathbb R^{n_d})$, { and the input \eqref{ua2} is  piecewise continuous  and Borel measurable on $[0,\infty)$, by the same argument as in Property~1 of Theorem~\ref{theo1}.}
Property 1 is obtained.

2) According to Lemmas \ref{lemdmd2}, \ref{lemada2}, we directly obtain the property 2) that the parameter identification $\hat{\Theta}(t)$ achieves the true value in finite time $t_f$.\par

3) Combining Lemma \ref{lemerr} with \eqref{p} and \eqref{tilp} from Assumption \ref{assumh3} yields the bound of $\hat{\eta}_2(t)$ as
\begin{align}
	|\hat{\eta}_2(t)|\le 2\max\{\|\tilde{K}e^{\mathcal{D}\mathcal{A}}\|,|\bar{\gamma}(1,\hat{\Theta})|,\xi_e\}\Omega_e\le \hat\rho_e(t).\label{tileta2}
\end{align}
As established in Lemma \ref{lemerr} via \eqref{omee2}, $\rho_e(t)$ converges exponentially to zero. This implies that $\hat\rho_e(t)$, as defined in \eqref{rho2}, and $w(1,t)$, as given by \eqref{w1t}, also converge exponentially to zero.  Following the same process in the proof of the property 3) in Theorem \ref{theo1},   this property  is obtained.\par

4)By \eqref{ua2}, we have
$\vartheta_0\hat{\eta}_1(t)
=\vartheta_0(U_a(t)-U(t))\ge0$
for $0\le t<t_f$.
 Applying the characteristic relation \(w(0,t)=w(1,t-D)\) with \eqref{tar2}, \eqref{w1t}, we have $b\vartheta(t)w(0,t)\ge0$ for $D\le t<D+t_f$. For $t\ge t_f$,  it follows from \eqref{w1t},  \eqref{tileta2} that
$
	\vartheta_0 w(1,t)=\hat\rho_e(t)+\vartheta_0\hat{\eta}_2(t)\ge \hat\rho_e(t)-|\hat{\eta}_2(t)|\ge0.
$ 
The same characteristic relation further gives $b\vartheta(t)w(0,t)\ge0$ for all $t\ge D+t_f$. Consequently, we conclude that $b\vartheta w(0,t)\ge0$ for all $t\in[D,\infty)$.
The selection of the design parameters $k_i$ in \eqref{kcon3}  defines a stricter subset of the conditions in \eqref{k}, which strictly guarantees the nonnegativity of the initial barrier function values, i.e., $h_{i}\left(\underline{z}_{i}\left(D\right), D\right)$. Thus, the result of Lemma \ref{lemmasafe2} remains valid. Following the same procedure used to prove property 4) in Theorem \ref{theo1}, we obtain this property.
The proof of this theorem is complete.

\end{proof}

\section{Application in Safe Vehicle Platooning}\label{secexample}

\begin{figure}[t]
	\centering	
	\includegraphics[width=1\linewidth]{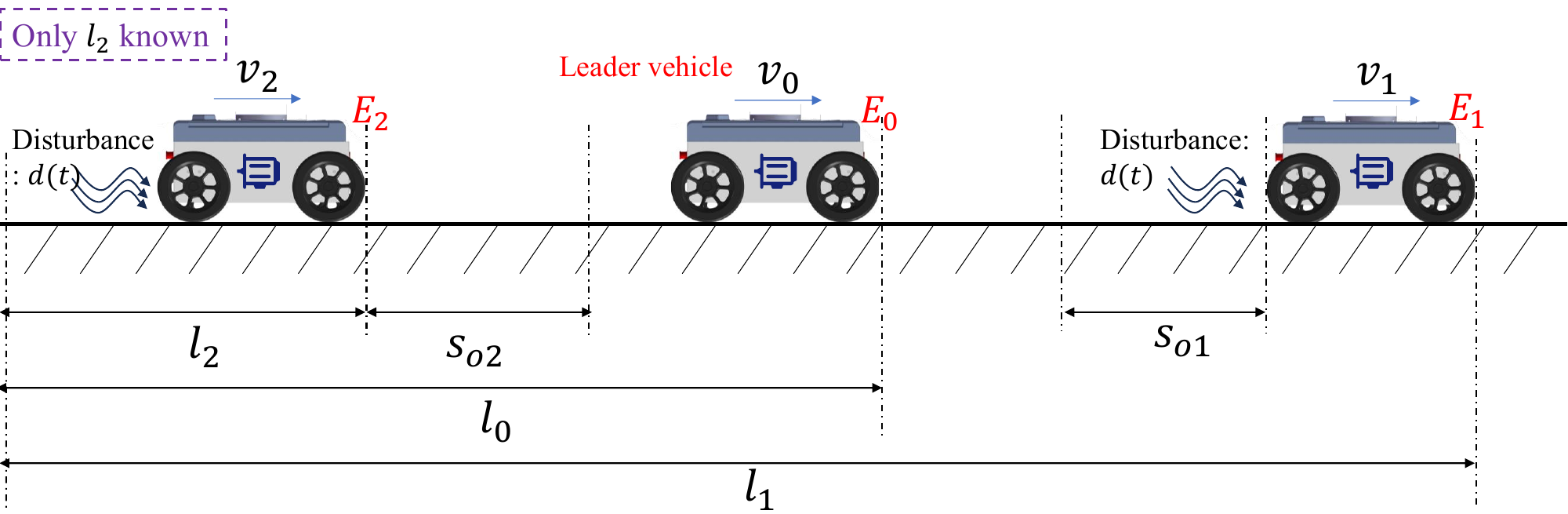}
	\caption{ Vehicle platooning with leader $E_0$, and the controlled $E_i,i=1,2$, where the safe distances to be maintained are $s_{oi}$.}
	\label{fig:0}
\end{figure}
\subsection{Physical model}
In this section, we validate the effectiveness of the designed controller through a practical application of safe vehicle platooning, depicted in Fig. \ref{fig:0}, which consists of three vehicles, where the middle vehicle, $E_0$, serves as a reference, traveling at a known speed $v_0(t)$. The control actions are applied to the other two vehicles: the preceding vehicle, $E_1$, and the following vehicle, $E_2$.
Each vehicle $E_i$ is subject to an unknown input delay $D_i$ between the control command $F_i(t)$ and the actual drive force $F_{w,i}(t)$, accounting for actuator dynamics and communication latency, i.e.,
$
	{F_w}_{i}(t)=F_i(t-D_i).\label{eq:delaysim}
$
According to Newtonian mechanics, the longitudinal dynamics of the $i$-th vehicle are described by
\begin{align}
	\frac{dl_i}{dt}=v_i(t)+d_1(t),\,
	M_{i}\frac{dv_i}{dt}={F_w}_{i}(t)-F_r+d_2(t),\label{veh1}
\end{align}	
where $M_i$ is the mass of vehicle $E_i$,  $l_i(t)$ and $v_i(t)$ are its  displacement and  speed respectively, and $F_r=f_1 v_i(t)$ represents the overall resistance force. The exogenous signal $d(t)=(d_1(t),d_2(t))^\top$ is introduced to capture the influence of external disturbances, such as wind gusts and road slope, as well as unmodeled dynamics and modeling errors.

 Considering  the disturbance $d(t)$  and target trajectory $r_i(t)$  are both generated by an exosystem $\dot{V}(t) = S V(t)$ with $S=\operatorname{diag}(S_d,S_r)$ and $V_i(t)=(V_d(t),V_{r,i}(t))^\top$ according to \eqref{modd}--\eqref{vd},  and defining  the state vector as $X_i(t) = (l_i(t), v_i(t))^\top$, and the control input as $U_i(t) = F_i(t)$,  the vehicle dynamics system can be expressed in state-space form: 
\begin{equation}
	\dot{X}_i(t) = 
	\begin{pmatrix}
		0 & 1\\
		0 & a_{i}
	\end{pmatrix}	
	X_i(t)
	+ \bar{G}V_i(t) +
	\begin{pmatrix}
		0\\
		b_i
	\end{pmatrix}
	U_i(t-D_i).
	\label{model}
\end{equation}
 Here, the disturbance part of exosystem $V(t)$ is  generated by  $S_d=\begin{pmatrix}0&1; \beta_1&0\end{pmatrix}$,  ${P}_d=I_2$  recalling \eqref{modd}. This signal affects the vehicle dynamics through the input matrix $\bar{G}=G\bar{P}_d$, where $G=\begin{pmatrix}1&1; 0&1\end{pmatrix}$ and $\bar{P}_d=(P_d,0)$.
Besides, the specified reference trajectory for each vehicle $E_i$ is $r_i(t)=l_0(t)+(-1)^{i-1}s_{oi}$, where $s_{oi}$  denotes the preset target distance. The leader's displacement is given by $l_0(t)=4t-\cos(t)+11$. This trajectory is produced by the reference generator from \eqref{modr1}, defined by $
	S_r =(
		0 , 1 , 0 , 0;
		0 , 0 , 1 , 0 ;
		0 , -1 , 0 , 1;
		0 , 0 , 0 , 0
)
$, $
	P_r = \begin{pmatrix} 1 , 0 , 0 , 0 \end{pmatrix}$, with
	 the initial condition $V_{r,i}(0) = (10+(-1)^{i-1}s_{oi}, 4, 1, 4)^\top$.

The model parameters are not precisely known which arise from various operational factors, including variations in vehicle load and unpredictable network latencies. We assume that these unknown parameters are bounded within known intervals, as shown in Table \ref{tab:params}.
{The uncertainty bounds in Table \ref{tab:params} can be
determined from prior knowledge of the physical system.
Specifically, the bounds on the vehicle mass $M_i$ and
resistance coefficient $f_1$ can be obtained from manufacturer
specifications or prior measurements, which in turn provide
the bounds on $a_i$ and $b_i$.
The delay bounds can be estimated from the characteristics
of the communication and actuator systems,
 and the bound
on $\beta_1$ can be determined from the known frequency
range of the disturbance.}

\subsection{Control objective}
The objective is to enforce a safe platoon formation by regulating the distances between adjacent vehicles, defined as $s_1(t) = l_1(t) - l_0(t)$ and $s_2(t) = l_0(t) - l_2(t)$, to their designated safe setpoints, $s_{o1}$ and $s_{o2}$, without breaching the
safety constraint. If starting from an unsafe condition, i.e., $s_i(D) < s_{oi}$, the system is expected to return to the safe region in a finite time assigned by users, and subsequently converge to the specified setpoints without breaching safe constraints.

\begin{table}[tb]
		\renewcommand{\arraystretch}{1.05}
	\centering
	\caption{Parameters and Uncertainty Bounds}
	\label{tab:params}
	\begin{tabular}{>{\scriptsize}p{12em}<{\raggedright} >{\scriptsize}p{8em}<{\centering}>{\scriptsize}p{3em}<{\centering}}
		\hline\\[-2.5mm]  
		\small\textbf{Parameter} &\small \textbf{values} & \small\textbf{Bounds} \\ \\[-2.5mm]\hline
		Linear damping coefficient: $f_1$ & 5 (Ns/m) & -- \\
		Vehicle mass :$M_1$ & 5 (kg) & -- \\
		Vehicle mass:$M_2$ & 4 (kg) & -- \\
		Input delay in $E_1$: $D_1$ & 1.5 (s) & $[1, 3]$ \\
		Input delay in $E_2$: $D_2$ & 1.2 (s) & $[1, 3]$ \\
		Input parameters $b_i$ & $1/M_i$ 	& 	$[0.1, 1]$ \\
		system parameters $a_{i}$ & $-f_1/M_i$ & $[-2,0]$ \\
		Disturbance parameter: $\beta_{1}$ & $-1/4$  & $[-1 , 0]$ \\ 
		\hline
	\end{tabular}
\end{table}

\textbf{Front vehicle $E_1$}: 
The safety constraint for the front vehicle  is formulated as the barrier function:
$
	\mathcal{H}_1(e_1(t),t) =X_{1,1}(t)- r_1(t)  = e_1(t) \ge 0.
$
For this vehicle,  all states, including disturbance states, are assumed to be available to test the performance of the control design in Sec. \ref{secfull}.

\textbf{Rear vehicle $E_2$}: 
The safety constraint  is formulated as the barrier function:
$
	\mathcal{H}_2(e_2(t),t) =r_2(t)-Y(t)   = -e_2(t) \ge 0.
$
 It is obtained from \eqref{var} that $\vartheta_2=-1$. For this vehicle, only the output state $Y(t)=x_{2,1}$ is available, to test the performance of the control design in Sec. \ref{secout}.

\subsection{Controller design}
The control gains are selected as $k_1=2$ and $k_2=1 $ according to \eqref{k}, \eqref{hk}, \eqref{kcon}, \eqref{kcon3}, with the initial distances $x_{1,1}(0)=17$, $x_{2,1}(0)=5$  and the initial velocities of $x_{i,2}(0)=5$ m/s. The initial value of the disturbance is $V_d(0)=(1,0)^\top$.	
For identification, we take $\tilde n=4$ and collect $2\tilde n=8$ snapshots with the sampling time $T_d=0.12$ s. The BaLSI update period is  $T_a=3$s.
Based on the derived controller \eqref{u}, the  control gains  for vehicles $E_i$ are calculated as:
$
	K_i=\frac{1}{b_i}(0,a_{i})$,$ \quad {G_0}_i=\frac{1}{b_i}(g_2+g_1S-\bar P_rS^2)
$,
$f_i=\frac{\vartheta_i}{b_i}((k_1+k_2)Z_{2,i}+k_1k_2Z_{1,i})$
\subsubsection{Case1:	Within safe region initially}

To verify the safety and convergence properties, the target distances are set as the smaller value $s_{o1}=0.5$ and $s_{o2}=0.5$. According to \eqref{u}, the nominal controller for each vehicle $E_i$ is
		\begin{align}
			&U_i(t)=-\frac{1}{b_i}\big[k_1k_2x_{i,1}(t+D_i)+(k_1+k_2+a_i)x_{i,2}(t+D_i)\notag\\
			&+\big((k_1+k_2)g_1+g_2+g_1S\notag\\
			&-(k_1k_2\bar{P}_r+(k_1+k_2)\bar{P}_rS+\bar{P}_rS^2)\big)e^{D_iS}V_i(t)\big]\label{usim}
		\end{align} 	
with the prediction $X_i(t+D_i)$ given by \eqref{predx}.
Considering the system parameter uncertainties $\Theta=\{a_1,\beta_1,b_1,D_1\}$, according to \eqref{ua}, the data-driven safe
 controller is derived from the nominal  controller \eqref{usim} based on the control design in Sec. \ref{seccon2}. 
 Besides,  for vehicle $E_2$, only the output state $Y(t)=x_{2,1}(t)$ is measurable and the system parameters $\Theta=\{a_2,\beta_1,D_2\}$ are unknown,  thus  the observer \eqref{ox}--\eqref{oy} is implemented to estimate the full state $X_2(t),V_d(t)$.  We use the estimated  states $\hat{X}_2(t), \hat{V}_d(t)$  to replace the  states $X_2(t), V_d(t)$ in \eqref{usim}, and  add the compensated term $\hat\rho_e(t;\hat{\Theta}(t))$ in \eqref{rho2}   to mitigate the effects of observer error and parameter estimation error, according to Sec. \ref{secconobe}. The observer gains are selected as
	$L_{v,0}=0$,
	$\mathcal{L}_v=(-0.6,15)^\top$,
	$L_{x,0}=0$,
	$\mathcal{L}_x=(9,10)^\top$
  recalling \eqref{lgain}. Besides, the known bound of unknown states is given by  $x_{2,2}\in[2,6]$ and $V_{d,i}\in[-2,2]$ satisfying  Assumption \ref{assumboundv}, and we take $\hat{x}_{2,2}=2$ and $\hat{V}_{d}=0$ as the initial values of the observer.

\subsubsection{Case2:	Within unsafe region initially}
Here, the target distances are set as larger values $s_{o1}=6$ and $s_{o2}=8$,  where the barrier function  $\mathcal{H}_i(e_i(D_i),D_i)<0$, which are initially unsafe. Considering ${\varsigma}_i(t)$ given by \eqref{sigma}, \eqref{sigma2}, with design parameters $\bar{t}=1.5,\epsilon=5$, there is an additional term $\bar{\varsigma}_i(t)=-\frac{1}{b_i\vartheta_i}(k_1k_2{\varsigma}_i(t+D_i)+(k_1+k_2)\dot{\varsigma}_i(t+D_i)+\ddot{\varsigma}_i(t+D_i))$ added in \eqref{usim}.

\begin{remark}\label{remvehiclecomputation}
	
	{The simulation in this paper is performed using the finite difference method with a time step $\diff t=0.001$ and a spatial step $\diff x=0.01$. For seeking the maximum or minimum of the candidate feedback recalling \eqref{ua}, \eqref{ua2} during parameter identification, which constitutes the main computational burden, we divide  the known parameter intervals  into finite grids with increments of $0.1$ to approximate it. We have further optimized the numerical implementation to reduce burden. First, the candidate feedback is affine in $1/b_1$, so only the two gain endpoints need to be evaluated. For the other unknown parameters, the coefficients multiplying the states and input history are precomputed once and stored in matrices, while the history integral is evaluated once per $(a_1,D_1)$ pair at each time step and reused across other parameter candidates. Matrix operations are then used to evaluate the candidate feedback values and select the extremum according to $\vartheta_0$. After parameter identification, the controller switches to the nominal feedback law, substantially reducing the computational burden.}

	{To examine the robustness of the proposed method to measurement noise, we have added noisy-data simulations
	for each vehicle.
During the identification process, independent zero-mean
Gaussian white noise is added to the measured data at
each  simulation step, with standard deviations
$\rho_1=1.15\times10^{-6}$ for $E_1$ and
$\rho_2=3.81\times10^{-10}$ for $E_2$.
				The noisy measurements are used for Krylov DMD
				identification, BaLSI, and the identified parameters are subsequently
				used in the controller and observer.}
\end{remark}

\subsection{Simulation results}
\begin{figure}[!t]
	\centering
	\includegraphics[width=0.98\columnwidth]{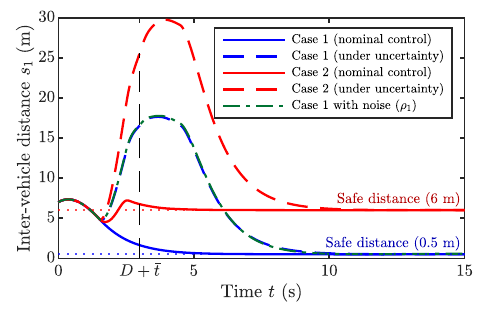}
	\caption{Distance between vehicles $E_1$ and $E_0$.}
	\label{fig1:spacing}
	\end{figure}
\begin{figure}[!t]
  \centering
  \begin{subfigure}[t]{0.5\columnwidth}
    \centering
    \includegraphics[width=\linewidth]{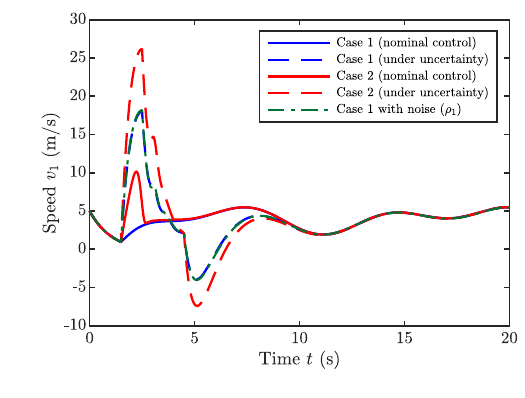}
    \caption{Speed $v_1(t)$}
    \label{fig1:speed}
  \end{subfigure}\hfill
  \begin{subfigure}[t]{0.5\columnwidth}
    \centering
    \includegraphics[width=\linewidth]{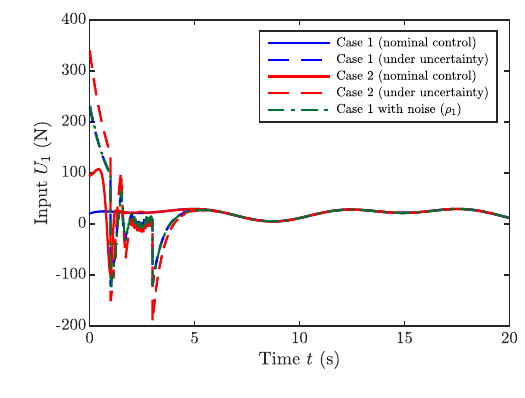}
    \caption{Input $U_1(t)$}
    \label{fig1:input}
  \end{subfigure}
  \caption{Speed and Control force of vehicle $E_1$.} 
  \label{fig1:speed-input}
\end{figure}
\begin{table}[!t]
\centering
\caption{Parameter estimates.}
\label{tab:identification}
\footnotesize
\setlength{\tabcolsep}{3pt}
\renewcommand{\arraystretch}{1.25}
\begin{tabular*}{\columnwidth}{@{\extracolsep{\fill}}ccccc@{}}
\hline
{\small\rule[-1ex]{0pt}{3.5ex}\textbf{Vehicle}} & {\small\textbf{Parameter}} & {\small\textbf{True value}} & {\small\textbf{Noiseless}} & {\small\textbf{Noisy}}\\
\hline
\multirow[c]{4}{*}{$E_1$} & $a_1$ & $-1$ & $-1.00$ & $-0.9991$\\
 & $\beta_1$ & $-0.25$ & $-0.25$ & $-0.2610$\\
 & $b_1$ & $0.2$ & $0.20$ & $0.1999$\\
 & $D_1$ & $1.5$ & $1.50$ & $1.5000$\\
\hline
\multirow[c]{3}{*}{$E_2$} & $a_2$ & $-1.25$ & $-1.25$ & $-1.2480$\\
 & $\beta_1$ & $-0.25$ & $-0.25$ & $-0.2490$\\
 & $D_2$ & $1.2$ & $1.20$ & $1.2000$\\
\hline
\end{tabular*}
\end{table}

The simulation results of vehicle $E_1$ are shown in Figs. \ref{fig1:spacing}--\ref{fig1:speed-input}. The blue and red lines correspond to Cases 1 and 2, respectively, while the solid and dashed lines denote the nominal control \eqref{u} and the proposed control \eqref{ua}.  The green dash-dotted lines represent the results	of Case 1 with measurement noise for both vehicles. As shown in Fig. \ref{fig1:spacing}, the inter-vehicle distance converges to the desired values $s_{o1}=0.5$ m and $s_{o1}=6$ m in Cases 1 and 2, respectively. Moreover, in Case 1, the safe distance is ensured all the time. In Case 2, from an unsafe vehicle distance at $t=D$,  vehicle distance is driven back and maintained in the safe distance within the predefined time $D+\bar{t}$. The vehicle's speed and control input are illustrated in Figs. \ref{fig1:speed}, \ref{fig1:input}, respectively.

The simulation results of vehicle $E_2$ are shown in Figs. \ref{fig2:spacing}-- \ref{fig2:obs}. As illustrated by the blue line in Fig. \ref{fig2:spacing}, the inter-vehicle distance converges to the desired safe distance $s_{o2}=0.5$ m without breaching the safety constraint in Case 1.  The red line in Fig. \ref{fig2:spacing} for Case 2 shows that, starting from an unsafe initial distance, the inter-vehicle distance is regulated to the safe region by the preset time $D+\bar{t}$, and then safely converges to the desired distance $s_{o2}=8$ m. The results of the speed $v_2(t)$ and control force $U_2(t)$ are illustrated in Fig. \ref{fig2:speed-input}.  The observer successfully estimates the true system states, as shown in  Fig. \ref{fig2:obsstates}, and the disturbance states, as shown in Fig. \ref{fig2:obsd1}.

Tab.~\ref{tab:identification} presents the true parameter
	values and their estimates obtained with and without
	measurement noise. Specifically, the system and disturbance
	coefficients are identified at $t=\underline D$, whereas
	the delay and input gain are identified at
	$t=t_f=3\,\mathrm{s}$.

\begin{figure}[!t]
	\centering
	\includegraphics[width=0.98\columnwidth]{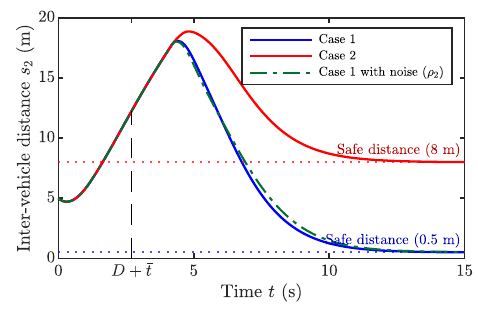}
	\caption{Distance between vehicles $E_2$ and $E_0$.}
	\label{fig2:spacing}
	\end{figure}
\begin{figure}[!t]
  \centering
  \begin{subfigure}[t]{0.49\columnwidth}
    \centering
    \includegraphics[width=\linewidth]{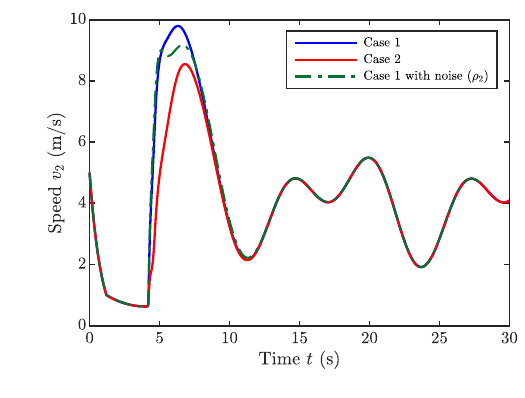}
    \caption{Speed $v_2(t)$}
    \label{fig2:speed}
  \end{subfigure}\hfill
  \begin{subfigure}[t]{0.49\columnwidth}
    \centering
    \includegraphics[width=\linewidth]{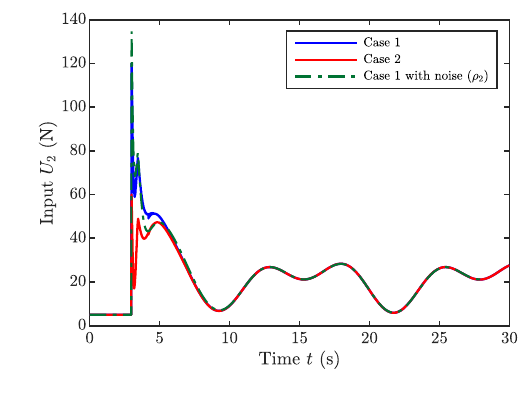}
    \caption{Input $U_2(t)$}
    \label{fig2:input}
  \end{subfigure}
  \caption{Speed and Control force of vehicle $E_2$.}
  \label{fig2:speed-input}
\end{figure}
\begin{figure}[!t]
  \centering
  \begin{subfigure}[t]{0.49\columnwidth}
    \centering
    \includegraphics[width=\linewidth]{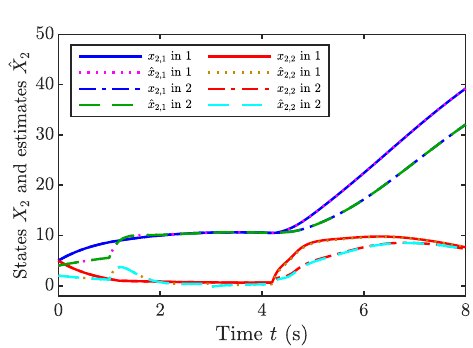}
    \caption{Estimates of $\hat{X}_2$}
    \label{fig2:obsstates}
  \end{subfigure}\hfill
  \begin{subfigure}[t]{0.49\columnwidth}
    \centering
    \includegraphics[width=\linewidth]{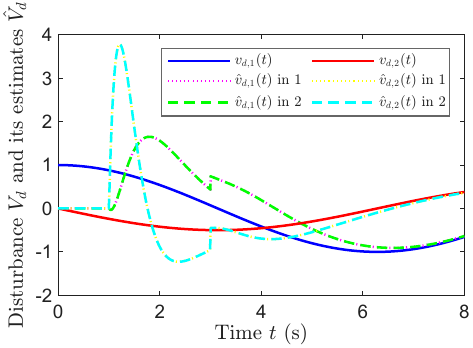}
    \caption{Estimates of $\hat{V}_d$}
    \label{fig2:obsd1}
  \end{subfigure}
  \caption{Estimates of states.} 
  \label{fig2:obs}
\end{figure}

\section{Conclusion and Future Work}\label{con}
This paper investigated the safe output regulation problem for strict-feedback linear systems, addressing challenges posed by substantial uncertainties, including unknown input delay, uncertain system dynamics, and external disturbances. A novel control design was developed by integrating a delay compensation mechanism, a Krylov DMD and BaLSI adaptive identifier, and a safety filter. For the more challenging scenario where only the furthest state as output is measurable,  the output-feedback control scheme was constructed based on the observer. The proposed control design guarantees that the output of the closed-loop system safely tracks the target trajectory, all states are bounded, and all unknown parameters are identified exactly in finite time. An application in safe vehicle platooning was presented to validate the theoretical results. {Future research will focus on observer-based reconstruction of the unmeasured actuator state, together with delay identification and safety enforcement.}

\appendices
\numberwithin{equation}{subsection}

\begin{appendix}
\subsection{The calculation of the PDE transformation \eqref{wxt} and its inverse transformation} \label{pdetrans}

The transformation \eqref{wxt} has the kernel representation $w(x,t)=u(x,t)-\int_0^x q(x,y)u(y,t)\diff y-\gamma(x)X(t)-\bar\gamma(x)V(t)-\psi(x,t)$.
Substituting the time and spatial derivatives of the transformation \eqref{wxt} into \eqref{tar2}, and requiring the resulting identity to hold for all $u$, $X(t)$, and $V(t)$, yields the following conditions:
$
-q_x(x,y)=q_y(x,y),\label{cond1}  
q(x,0)=D\gamma(x)B,\label{cond2}
\dot{\gamma}(x)=D\gamma(x)A,\label{cond3}
\dot{\bar{\gamma}}(x)=D\bar{\gamma}(x)S+D\gamma(x)\bar{G},\label{cond4}
D\psi_t=\psi_x\label{cond5}.
$
Comparing \eqref{doth} with \eqref{tar1} yields the boundary condition
$
w(0,t)=u(0,t)+\frac{1}{\vartheta}f(\underline{z}_n(t),t)-\frac{1}{b}(\sigma_n+\bar{P}_rS^{n})V(t)+ KX(t)		
$.
By setting $x=0$ in \eqref{wxt} and comparing the result with the boundary condition above, we obtain the initial conditions
$
\gamma(0)=-K,\label{inicond1}
\bar\gamma(0)=\frac{1}{b}(\sigma_n+\bar{P}_rS^{n}),\label{inicond2}
\psi(0,t)=\frac{-1}{\vartheta(z_{1}(t),t)}f(\underline{z}_{n},t).\label{inicond3} 
$
Solving the ODEs $\gamma,\bar{\gamma}$ with the above initial conditions yields $\gamma(x)=-Ke^{DAx}$ and $\bar{\gamma}(x)$ in \eqref{bargam}.
Finally, considering the  general solution $Q(x-y)$  to $q(x,y)$, $\bar{\psi}(t-D+Dx)$ to $\psi$, we get 
\begin{align}
	q(x,y)&=-DKe^{DA(x-y)}B,\\
	\psi(x,t)&=-\frac{f\left(\underline{z}_n(t+Dx),t+Dx\right)}{\vartheta(z_1(t+Dx),t+Dx)}.\label{solpsi}
\end{align}

The inverse transformation is given by \eqref{inverw}.
Through similar calculations, its components are explicitly solved as $\bar{q}(x,y)=-DKe^{D(A-BK)(x-y)}B$, $\gamma_1(x)=-Ke^{D(A-BK)x}$, and $	\bar{\gamma}_1(x)=-G_0e^{DSx}+\int_{0}^{x}D\gamma_1(y)(\bar{G}-BG_0)e^{DS(x-y)}\diff y $. The remaining term, $\psi_1(x,t)$, is calculated based on $\gamma_1(x)$ and $\psi(x,t)$.

\subsection{The proof of Lemma \ref{lemmasafe2}}\label{appcbf}
Setting $t=D$  in \eqref{hi}, then one gets
\begin{align}
	h_{i+1}\left(\underline{z}_{i+1}\left(D\right), D\right)=&\sum_{j=1}^i \frac{\partial h_i}{\partial z_j} z_{j+1}\left(D\right)+\frac{\partial h_i}{\partial t}\notag\\
	&+k_i	h_{i}\left(\underline{z}_{i}\left(D\right), D\right).	\label{zd}
\end{align}
The above recursive formula \eqref{zd} is used to establish the inductive step : if $h_{i}\left(\underline{z}_{i}\left(D\right), D\right)>0$, then $h_{i+1}\left(\underline{z}_{i+1}\left(D\right), D\right)\ge0$ under the selection of parameters $k_i$ in \eqref{k}, \eqref{hk}. The base case, $h_1(z_1(D),D)>0$, is guaranteed by \eqref{h1}, \eqref{sigma}. It thus follows by induction that $h_{i}\left(\underline{z}_{i}\left(D\right), D\right)\ge0$ for all $i=1,\cdots,n$. Given these safe initial conditions and $w[t]\equiv0$ for $t\ge D$, the system structure described by \eqref{tar1}--\eqref{tar3} guarantees that non-negativity is maintained for all subsequent times, i.e., $h_{i}\left(\underline{z}_{i}\left(t\right), t\right)\ge0$ for all $t\in[D,\infty)$ and $i=1,\cdots,n$.

\subsection{The proof of Lemma \ref{lembound}} \label{appbound}
Considering the boundedness of  $\vartheta$ \eqref{var} concluded from Assumption \ref{assumh}, we know that $w[t]$, $H(t)$ are bounded all the time from the target system \eqref{tar1}--\eqref{tar3} utilizing the method of characteristics. Especially, the matrix $A_h$ is Hurwitz  with the above selection of $k_i$ \eqref{k} and $w[t]\equiv0$ for $t\ge D$, which further indicates that $|H(t)|$ is exponentially convergent to zero. From \eqref{h1}, \eqref{sigma}, Assumption \ref{assumh} and the boundedness of $h_1$, we obtain that $z_1(t)=e(t)$ is bounded. Differentiating  \eqref{hi} with respect to $z_{i-1}$, we have 
$\frac{\partial h_{i}}{\partial z_{i-1}}=\frac{\partial h_{i-1}}{\partial z_{i-2}}+\frac{\partial\vartheta}{\partial z_1}z_2+\frac{\partial^2 h_{i-1}}{\partial z_{i-1}\partial t}+k_{i-1}\frac{\partial h_{i-1}}{\partial z_{i-1}}
=\frac{\partial h_{i-1}}{\partial z_{i-2}}+\frac{\partial\vartheta}{\partial z_1}z_2+\frac{\partial^2 h_{1}}{\partial z_{1}\partial t}+k_{i-1}\frac{\partial h_{1}}{\partial z_{1}},\quad i\ge3$. Beginning from $\frac{\partial h_{2}}{\partial z_{1}} =\frac{\partial\vartheta}{\partial z_1}z_2+\frac{\partial^2 h_{1}}{\partial z_{1}\partial t}+k_{1}\frac{\partial h_{1}}{\partial z_{1}}$, we obtain that
$\frac{\partial h_{i}}{\partial z_{i-1}}=\frac{\partial h_1}{\partial z_1}\sum_{j=1}^{i-1}k_j+(i-1)(\frac{\partial\vartheta}{\partial z_1}z_2+\frac{\partial^2 h_1}{\partial z_1\partial t}),\quad\forall i\geq2$ from the induction step. Recursively, we obtain $\sum_{i=1}^n\sum_{j=1}^i|\frac{\partial h_i}{\partial z_j}|\leq\mathcal{M}(|\frac{\partial h_1}{\partial z_1}|+\sum_{j=1}^{n-1}|\frac{\partial h_1}{\partial z_1\partial t^{n-1}}|)$ for some positive constant $\mathcal{M}$. Since $\varsigma$ has continuous derivatives of all orders recalling \eqref{h1}, \eqref{sigma} and Assumption \ref{assumh}, it further indicates that all $\frac{\partial h_i}{\partial z_j},1\leq j\leq i\leq n,i,j\in\mathbb{N}$ are bounded. The boundedness of $Z(t)$ is obtained from the boundedness of $\frac{\partial h_i}{\partial z_j}$,$H(t)$. Next, we have the boundedness of $|X(t)|$ from \eqref{Z}, together with the boundedness of $V(t)$ from \eqref{modv}. Applying the Cauchy-Schwarz inequality to the inverse transformation \eqref{inverw} given in Appendix \ref{pdetrans}, we have that $u[t]$ is bounded. Especially, as defined in \eqref{sigma}, $\varsigma\equiv0$ is ensured after a finite time $t=D+\bar{t}$, regardless of which case applies. Thus, we know that $e(t)$ is exponentially convergent to zero from the exponential convergence to zero of $h_1(z_1(t),t)$, recalling \eqref{h1} and Assumption \ref{assumh}. The proof  is complete.

\subsection{The proof of Lemma \ref{lemfull1}}\label{appfull1}
It readily follows from \eqref{qo} that $\operatorname{rank} Q_o=\tilde{n}$. Since $(\tilde{A},I)$ is controllable and $\tilde{A}$ is cyclic from Assumption \ref{assumeigen}, we know that $(\tilde{A},\tilde{X}(0))$ is controllable for almost all $\tilde{X}(0)\in\mathbb{R}^{\tilde{n}}$. It indicates that $\operatorname{rank} Q_c=\tilde{n}$ given in \eqref{qc} with Assumption \ref{assumtd}. Applying Sylvester's rank inequality and Rank of product inequality to \eqref{qoc},  we know that $	\operatorname{rank}	H_{\tilde{n}}(\tilde{X}) =\tilde{n} $. The proof is complete.

\subsection{The proof of Lemma \ref{lemeigen}}\label{appf}
Applying \eqref{qoc} with \eqref{kh} yields that $Q_o Q_c f=-Q_o \tilde{A}_d^{\tilde{n}} \tilde{X}(0)$, which is equivalent to
\begin{equation}
	Q_o( Q_c f+\tilde{A}_d^{\tilde{n}} \tilde{X}(0))=0. \label{qoe}
\end{equation} 
Given that $\operatorname{rank} H_{\tilde{n}}(\tilde{X})=\tilde{n}$, the observability matrix $Q_o$ has full column rank, i.e., $\operatorname{rank}(Q_o) = \tilde{n}$ and $\operatorname{rank}Q_c ={\tilde{n}}$ as well. Therefore, $Q_o$ has a trivial null space,
such that $Q_c f+\tilde{A}_d^{\tilde{n}} \tilde{X}(0)=0$ only holds for \eqref{qoe} and has a unique solution because $\det Q_c\neq0$. Combining with \eqref{qc}, rearranging gives 
\begin{equation}
	P_F(\tilde{A}_d) \tilde{X}(0)=0,\label{pf}
\end{equation}
where $P_F(\tilde{A}_d)=(\tilde{A}_d^{\tilde{n}}+f_{\tilde{n}-1} \tilde{A}_d^{\tilde{n}-1}+\ldots+f_1 \tilde{A}_d+f_0 I)$ and $P_F(\lambda)=\det(\lambda I-F)$ denotes the characteristic polynomial of the companion matrix $F$ given in \eqref{kh}. 
According to Assumption \ref{assumeigen}, the eigenvectors $\tilde{\nu}_i$ are linearly independent, so the matrix $\tilde{V}$ is invertible.
Using the eigendecomposition of the system, $\tilde{A}_d = \tilde{V} \hat{\Lambda} \tilde{V}^{-1}$ and $\tilde{X}(0)=\tilde{V}\underline{\tilde{\varphi}}(\tilde{X}(0))$ from \eqref{tilx}, we can rewrite \eqref{pf} as
\begin{equation}
	\tilde{V} P_F(\hat{\Lambda}) \tilde{V}^{-1} \tilde{V} \underline{\tilde{\varphi}}(\tilde{X}(0)) = \tilde{V} P_F(\hat{\Lambda}) \underline{\tilde{\varphi}}(\tilde{X}(0)) = 0, \label{pf2}
\end{equation}
 where $\hat{\Lambda}=\operatorname{diag}(\hat{\lambda}_1,\cdots,\hat{\lambda}_{\tilde n})$.
Because $\det Q_c\neq0$ and $ \det \tilde{V} \neq0$, we know that $	\tilde\varphi_i(\tilde{X}(0))\neq0, i =1,\cdots, \tilde{n}$ by using 
the PBH eigenvector  test  and \eqref{eigenfun} , such that  only
$P_F(\hat{\Lambda}) = 0$ can hold in \eqref{pf2}. Therefore, the eigenvalues of $\tilde{A}_d$ are identical to the eigenvalues of $F$, i.e., 
$\hat{\lambda}_{i}=\mathrm{e}^{\tilde{\lambda}_{i}T_d}$, which gives $\tilde{\lambda}_i=\frac{\ln\hat{\lambda}_i}{T_d}$ for $i=1,\ldots,\tilde{n}$.
The Koopman eigenvalues are solved as $\tilde{\lambda}_i=\frac{\ln\hat{\lambda}_i}{T_d},\quad i=1,\ldots,\tilde{n} $.

Back to the Hankel matrix $H_{\tilde{n}}(\tilde{X})$ \eqref{qoc}, 
 we substitute the state decomposition \eqref{tilx} into \eqref{qc}  to obtain
$
H_{\tilde{n}}(\tilde{X})=Q_{o}\tilde{V}\left[\underline{\tilde{\varphi}}(\tilde{X}(0))\quad\hat{\Lambda}\underline{\tilde{\varphi}}(\tilde{X}(0))\quad\ldots\quad\hat{\Lambda}^{\tilde{n}-1}\underline{\tilde{\varphi}}(\tilde{X}(0))\right]
=Q_{o}\tilde{V}\mathrm{diag}(\tilde{\varphi}_{i}(\tilde{X}(0)))V_{F},
$
where $\mathrm{diag}(\tilde{\varphi}_{i}(\tilde{X}(0)))$ is the diagonal matrix of initial eigenfunction values, and $V_F=\begin{pmatrix}1_{\tilde{n}}&\operatorname{col}(\hat{\lambda}_i)&\ldots&\operatorname{col}(\hat{\lambda}_i^{\tilde{n}-1})\end{pmatrix}$. Especially, the Vandermonde matrix $V_F$ is the inverse modal matrix of the companion matrix $F$, which implies that $V_F\hat{\nu}_i=e_i$, where $e_i$ is the $i$-th standard basis vector.
Consequently, we can write
\begin{equation}
	H_{\tilde{n}}(\tilde{X})\hat{\nu}_{i}=Q_{o}\tilde{V}\mathrm{diag}(\tilde{\varphi}_{i}(\tilde{X}(0))){V_{F}\hat{\nu}_{i}}=\tilde{\varphi}_{i}(\tilde{X}(0))Q_{o}\tilde{\nu}_{i}.
\end{equation}
Since the first block-row of $Q_o\tilde{\nu}_i$ is $\tilde{\nu}_i$ by \eqref{qo}, and $\tilde\varphi_i(\tilde{X}(0))\neq0$, $i=1,\cdots,\tilde n$, we isolate the Koopman modes as
$	\tilde{\nu}_{i}=[I_{\tilde{n}}\quad 0]H_{\tilde{n}}(\tilde{X})\hat{\nu}_i,i=1,\cdots, \tilde{n}$.
The proof is complete.

\subsection{The proof of Lemma \ref{lemdmd}}\label{appdmd}
	Lemma \ref{lemeigen} establishes that we can compute the Koopman operator eigenpairs from the measurement data using DMD. For the linear system \eqref{tilA} and observable $\tilde{X}$ \eqref{tilx}, Koopman eigenvalues and modes are also the  eigenvalues and right eigenvectors of the system matrix $\tilde{A}$, respectively. Furthermore, Assumption \ref{assumeigen} states that $\tilde{A}$ is diagonalizable. Therefore, we exactly identify the system matrix $\tilde{A}$ from the eigendecomposition obtained in Lemma \ref{lemeigen} via \eqref{eq:tA}.
	 The estimated system parameters $\hat{\Theta}_1(\underline{D})=\{\mathcal{A},\mathcal{S}_d,\bar{\mathcal{G}}\}$ are then obtained by partitioning the reconstructed matrix $\tilde{A}$ according to its block structure defined in \eqref{tilA}, i.e., \eqref{matha}--\eqref{mathg}.

\subsection{Proof of Lemma \ref{lemada}}\label{appbalsi}

We briefly describe the design procedure of \eqref{est1} below. It is obtained from \eqref{org11}--\eqref{org13} that $f_n(t)=	Dq_n(t), 
f_b(t)=bq_b(t)$,
based on which the quadratic function $\Xi_{i,n}:\mathbb{R}^2\rightarrow\mathbb{R}_+$ is constructed as
$
	\Xi_{i,{n}}(\ell)=\int_{\mu_{i+1}}^{t_{i+1}}[(f_{{n}}-\ell_{D}q_{{n}})^{2}+(f_b-\ell_{b}q_b)^{2}]dt,\label{Xi}
$
for $i\in \mathbb{N}$, where $\ell=(\ell_D,\ell_b) ^\top$. It follows that the true parameter vector $\Theta_2$ is  a global  minimizer, such that
the gradient of $\Xi_{i,n}(\ell)$ vanishes here according to Fermat's theorem. Thus,  setting the partial derivatives  of $\Xi_{i,{n}}$ with respect to $\ell_D$, $\ell_b$ equal to zero  yields that $ \bar F_n(t_{i+1},\mu_{i+1})=\bar{Q}_n(t_{i+1},\mu_{i+1})\Theta_2 $ for each mode $n \in \mathbb{N}$ and update step $i \in \mathbb{N}$. The estimator \eqref{est1} is then derived from such a batch least-squares solution. More details can be found in \cite{adasafe1}.

{
	Following the approach in [\citen{adasafe1}, Lemma 4],
we first determine the estimates at each update and then prove finite-time identification and preservation of the true values. For the batch $[\mu_{i+1},t_{i+1}]$, the feasible set
	in \eqref{est1} is defined as
$\mathcal E_i=\{\ell\in\bar\Theta_2:\bar F_n(t_{i+1},\mu_{i+1})=\bar Q_n(t_{i+1},\mu_{i+1})\ell,\ n\in\mathbb N\}$.
	Multiplying \eqref{org10} by $-\sin(n\pi x)$ and integrating
	over space and time, together with $u[0]=0$, yields
	$f_n=Dq_n$. Similarly, integration of the last equation
	in \eqref{org11} gives $f_b=bq_b$.
	Consequently, we obtain $F_n=DQ_n$ and $F_b=bQ_b$,
	which implies $(D,b)\in\mathcal E_i$ at every update.
}

{
The diagonal structure of $\bar Q_n$ implies that every $\ell\in\mathcal E_i$ satisfies $Q_n(\ell_D-D)=0$ for all $n\in\mathbb N$ and $Q_b(\ell_b-b)=0$.
	If $Q_n>0$ for at least one $n$, the delay coordinate
	of $\mathcal E_i$ is uniquely determined as $D$.
	Otherwise, it can take any value in
	$[\underline D,\overline D]$.
	Similarly, the gain coordinate is uniquely determined as $b$
	when $Q_b>0$, and can take any value in
	$[\underline b,\overline b]$ when $Q_b=0$.
	Since $\bar\Theta_2$ is a Cartesian product and the objective
	in \eqref{est1} is the sum of two squared coordinate distances,
	each informative coordinate is updated to its true value,
	whereas each uninformative coordinate remains at its
	previous estimate.}

{
	We next establish the conditions under which the regression
	data are informative. As in [\citen{delay9}, Lemma 1],
	$Q_n=0$ for all $n\in\mathbb N_+$ is equivalent to
	$u[t]=0$ throughout the batch.
	Indeed, $Q_n=0$ implies that $q_n(t)=0$ almost everywhere
	on the batch. Since $q_n$ is continuous, this equality
	holds throughout the batch.
	By the completeness of
	$\{\sqrt{2}\sin(n\pi x):n\in\mathbb N_+\}$
	in $L^2(0,1)$, we obtain $u[t]=0$ for every time
	in the batch. The converse follows directly from \eqref{gn}.
	For the gain regression, a nonzero $u(0,t)$ on a subinterval
	implies that $q_b$ cannot vanish identically, since
	$\dot q_b(t)=u(0,t)$ almost everywhere.
	It follows that $Q_b>0$.
}

{
	For $t_f=t_j$ defined in \eqref{tf}, the piecewise continuity
	and right-continuity of the arriving input imply that there
	exists a subinterval
	$(\tau_1,\tau_2)\subset(t_{j-1},t_j)$
	on which $u(0,t)\ne0$.
	Since $\tilde N\ge1$, \eqref{mu} gives
	$\mu_j\le t_{j-1}$, and this subinterval is included
	in the batch used at $t_j$.
	Consequently, $Q_b>0$.
	Moreover, the characteristic relation
	$u(x,t)=u(0,t+Dx)$ implies that, for sufficiently small
	$x>0$ satisfying $Dx<\tau_2-\tau_1$,
	the profile $u(x,t)$ is nonzero for
	$t\in(\tau_1,\tau_2-Dx)$.
	Hence $u[t]$ is not identically zero on this batch,
	and the preceding argument gives $Q_n>0$ for at least
	one $n\in\mathbb N_+$.
	Therefore, $\mathcal E_{j-1}=\{(D,b)\}$,
	and \eqref{est1} yields
	$\hat\Theta_2(t_f)=\Theta_2$.
}

{
	Once $\hat\Theta_2(t_i)=\Theta_2$, the true parameter pair
	remains feasible at every subsequent update and gives
	zero value of the objective in \eqref{est1}.
	It is therefore selected again, even when the subsequent
	batch is uninformative.
	Since the update times satisfy $t_i=t_0+iT_a\to\infty$
	and the estimates remain constant  between consecutive parameter update times,
	we conclude that
	$\hat\Theta_2(t)=\Theta_2$ for all $t\ge t_f$.
	If the input is identically zero, both regressions vanish.
	In this case, a small finite-duration excitation in an
	admissible safe input direction can be introduced to generate
	an informative batch, as discussed in \cite{adasafe1}.
}

\subsection{The proof of Lemma \ref{lemfull2}}\label{appfull2}
For all $\tilde{\lambda}\in\operatorname{eig}(S_d)$, it is obtained from Assumption \ref{assumeigen} that the matrix $A-\tilde{\lambda} I $ is nonsingular. It allows the following rank calculation
	\begin{align}
		&\operatorname{rank}\begin{pmatrix}A-\tilde{\lambda} I&GP_d\\C&0\end{pmatrix}\notag\\
		&=\operatorname{rank}\begin{pmatrix}I&0\\-C(A-\tilde{\lambda} I)^{-1}&I\end{pmatrix}\begin{pmatrix}A-\tilde{\lambda} I&GP_d\\C&0\end{pmatrix}\notag\\
		&=\operatorname{rank}\begin{pmatrix}A-\tilde{\lambda} I&GP_d\\0&P_C(\tilde{\lambda})\end{pmatrix}
	\end{align}	
	where $P_C(\tilde{\lambda})=-C(A-\tilde{\lambda} I)^{-1}GP_d\in\mathbb{C}^{1\times n_d}$.	It yields with \eqref{rank1} that $\operatorname{rank} P_C(\tilde{\lambda}) =1,\forall\tilde{\lambda}\in\operatorname{eig}(S_d)$.  Since $C=[1,0,\cdots]$, we know that  the first element of the vector $P_C(\tilde{\lambda})$ is not zero for all $\tilde{\lambda}\in\operatorname{eig}(S_d)$. The eigenvectors of the companion matrix $S_d$ \eqref{a2} are of the form
	$
	v_{d,i}=\begin{pmatrix}1,\tilde{\lambda}_i,	\tilde{\lambda}_i^2,\cdots,	{\tilde{\lambda}_i^{n_d-1}}\end{pmatrix}^\top   ,i=1,\cdots,n_d.
	$
	It follows that $P_C(\tilde{\lambda}_i)v_{d,i}\neq0$, $i=1,\cdots,n_d $. We further know from \eqref{nu}, \eqref{nux} that $\bar{C}\tilde{\nu}_i\neq0 $, for $i=1,\cdots,n_d $. 
	The observability of $(C,A)$ and the PBH test ensure $C\nu_i \neq 0$, which implies $\bar{C}\tilde{\nu}_i = C\nu_i \neq 0$ for $ i = n_d+1, \dots, \tilde{n}$. Since $\bar{C}\tilde{\nu}_i \neq 0$ for all $i = 1, \dots, \tilde{n}$, the PBH test confirms that  $(\bar{C}, \tilde{A})$ is observable. It follows from Assumption \ref{assumtd} that  $(\bar{C}, \tilde{A}_d)$ is observable as well. Consequently, the observability matrix $Q_o$ in \eqref{qo2} has full rank, i.e., $\operatorname{rank}(Q_o) = \tilde{n}$.

	By the same argument as in the proof of Lemma \ref{lemfull1}, we have that $\operatorname{rank}Q_c=\tilde{n}$.  Since both $Q_o$ and $Q_c$ defined in \eqref{qo2}, \eqref{qc} are nonsingular, their product, the Hankel matrix $	H_{\tilde{n}}( Y)$ \eqref{hankely}, is also nonsingular and thus has full rank, i.e., \eqref{ranhy}.

\subsection{The proof of Lemma \ref{lemdmd2}}\label{appdmd2}

 Following the same process in the proof of Lemma \ref{lemeigen}, we directly obtain the result \eqref{tillam}.

According to Assumption \ref{assumeigen} that all eigenvalues $\tilde{\lambda}_i$ are distinct, we can construct a Vandermonde matrix \eqref{eq:vandermonde_matrices} from them, serving as the eigenvector matrix $V_A$ for the companion matrix ${A}$. 

{By Assumption \ref{assumspecsep}, all identified nonzero
imaginary-axis eigenvalues belong to $S_d$, and  those outside
the imaginary axis belong to $A$. Since $S_d$ is real and its
eigenvalues are simple and lie on the imaginary axis, the nonzero
eigenvalues occur in complex-conjugate pairs, and zero is an
eigenvalue of $S_d$ if and only if $n_d$ is odd.
We exclude common eigenvalues of $A$
and $S_d$ recalling Assumption \ref{assumeigen}. Thus, an identified zero eigenvalue belongs to $S_d$
when $n_d$ is odd and to $A$ when $n_d$ is even.
After relabeling, we have
$\{\tilde{\lambda}_1,\ldots,\tilde{\lambda}_{n_d}\}
=\operatorname{eig}(S_d)$ and
$\{\tilde{\lambda}_{n_d+1},\ldots,\tilde{\lambda}_{\tilde n}\}
=\operatorname{eig}(A)$.
Since all eigenvalues are distinct, the two Vandermonde matrices
in \eqref{eq:vandermonde_matrices} are nonsingular.
By the companion structure of $A$ and $S_d$ in \eqref{a2},
these Vandermonde matrices are their respective eigenvector matrices.
Hence,}
\begin{equation}
{\color{blue}
\mathcal{A}=V_A\Lambda_AV_A^{-1},\qquad
\mathcal{S}_d=V_d\Lambda_dV_d^{-1},
}
\end{equation}
{which proves $\{\mathcal{A},\mathcal{S}_d\}=\{A,S_d\}$.}

\subsection{Proof of Lemma \ref{lemada2}}\label{appbalsi2}

{The delay-identification argument in the proof of
	Lemma \ref{lemada} in Appendix \ref{appbalsi} establishes
	$F_n=DQ_n$. At $t_f$ defined by \eqref{tf2}, the retained
	batch contains a nonzero actuator profile, which implies
	$Q_n>0$ for at least one sine mode.
	Thus, every feasible $\ell$ in \eqref{est2} satisfies
	$Q_n(\ell-D)=0$, and we obtain $\hat D(t_f)=D$.
	At each subsequent update, the true delay $D$ remains
	feasible and yields zero distance from the preceding estimate.
	Therefore, we have $\hat D(t)=D$ for all
	$t\ge t_f$.}

\subsection{The proof of Lemma \ref{lemerr}}\label{appobserver}

The solution to the observer error ODE \eqref{oerr} can be expressed piecewise as
\begin{equation}
	e_o(t)=\begin{cases}
		e^{L_0t}e_o(0)+\int_0^t e^{L_0(t-\tau)}\Delta{e}(\tau)\diff \tau, & t\in[0,\underline{D})\\
		e^{\mathcal{L}(t-\underline{D})}e_o(\underline{D})+\int_{\underline{D}}^t e^{\mathcal{L}(t-\tau)}\Delta{e}(\tau)\diff \tau, & t\in[\underline{D},t_f)\\
		e^{\mathcal{L}(t-t_f)}e_o(t_f). & t\in[t_f,\infty)
	\end{cases}\label{oerrsol}
\end{equation}
Before the parameter estimates converge at $t_f$, the error dynamics are perturbed by the term $\Delta e(\tau)$. The magnitude of this perturbation is bounded, which can be derived from the known bounds of system parameters $A, S_d, D$ according to Assumptions \ref{assumboundd}, \ref{assum:theta_bounds}. Using the induced 2-norm, these bounds are defined as
$
	\|\Delta S_d(t)\| \le \delta_{S_d},  \|\Delta A(t)\| \le \delta_A,  |\Delta U| \le \delta_U(t), 
$

The constants $\delta_{S_d}$ and $\delta_A$ are positive. Under  the initial condition \eqref{u0}, the function $\delta_U(t)=2\sup_{-\overline D\le s\le t-\underline D}|U_a(s)|$ satisfies $|\Delta U(t)|\le\delta_U(t)$ because $D,\hat{\mathcal D}(t)\in[\underline D,\overline D]$. This bound depends only on  $U_a(s)$ for $s\in[-\overline D,t-\underline D]$. The following estimates hold on each compact interval on which the closed-loop solution exists.
By the triangle inequality, the upper bound of the perturbation term $\Delta e(t)$ for $t\in[0,\underline{D})$ is derived as 
\begin{align}
	\|\Delta e(t)\| &= \left\| \begin{pmatrix} \Delta S_d(t)\hat{V}_d(t) \\ \Delta A(t)\hat{X}(t) + B\Delta U \end{pmatrix} \right\| \nonumber \\
	&\le \|\Delta S_d(t)\||\hat{V}_d(t)| + \|\Delta A(t)\||\hat{X}(t)| + |B||\Delta U| \nonumber \\
	&\le \delta_{S_d} |\hat{V}_d(t)| + \delta_A |\hat{X}(t)| + b \delta_U(t) := \Gamma_1(t)
\end{align}
 The perturbation $\Delta e(t)$ vanishes for $t\ge t_f$.
Next, we calculate the explicit upper bounds of the observer errors. First, for $t\in[0,\underline{D})$, utilizing Lemma \ref{lemerr1}, we bound the matrix $e^{L_0 t}$ over the uncertainty set ${\Theta}_1$ defined in \eqref{theta} as
$
	\|e^{L_0 t}\| \le M_{L_0}(t)e^{-\delta_{L_0} t}\le \bar{M}_{L_0}e^{-\bar\delta_{L_0} t},\label{equlemerr1}
$
where $\bar\delta_{L_0}=-\max_{\theta_1\in \bar{\Theta}_1} \max\Re (\operatorname{eig} (L_0(\theta_1)))$ is the worst-case decay rate over all possible parameters in $\bar{\Theta}_1$.

The constant $\bar M_{L_0}=\max_{\theta_1\in\bar\Theta_1,\,0\le s\le\underline D}e^{\bar\delta_{L_0}s}\|e^{L_0(\theta_1)s}\|$ is finite by continuity and compactness.
Accordingly, for the interval $t \in [0, \underline{D})$, we establish
the upper bound for the error state norm $|e_o(t)|$ as
\begin{align}
	&|e_o(t)| \le  M_0 \bar{M}_{L_0}e^{-\bar\delta_{L_0} t}\notag\\
	&+ \int_0^t \bar{M}_{L_0}e^{-\bar\delta_{L_0} (t-\tau)} \Gamma_1(\tau) d\tau:=\Gamma_{e,1}(t),\label{eq:game1}
\end{align}

The constant $M_0$ is chosen to satisfy $M_0\ge\max_{\zeta_0\in\bar\zeta_0}|\zeta_0-\operatorname{col}(\hat V_d(0),\hat X(0))|$, which gives $|e_o(0)|\le M_0$. The set $\bar\zeta_0$ contains the admissible initial vectors $\operatorname{col}(V_d(0),X(0))$ under Assumption \ref{assumboundv}.
 Next, recalling Lemma \ref{lemdmd2}, we know that the error  terms $\Delta{S}_d(t)=0$, $\Delta A(t)=0$ after $t \ge\underline{D}$. It implies that the parameter estimation error term $\Delta{e}(t)$ in \eqref{oerr} becomes $(0\quad B \Delta U(t))^\top$ after $t> \underline{D}$, and thus the upper bound of the perturbation term $\Delta e(t)$ becomes
$
	\|\Delta e(t)\| \le b \delta_U(t) := \Gamma_2(t), t\in[\underline{D},t_f).
$
The upper bound of the whole error state $|e_o(t)|$ for $t\in[\underline{D},t_f)$ is derived as
\begin{align}
	&|e_o(t)| \le  \Gamma_{e,1}(\underline{D}) M_{\mathcal{L}}(t-\underline{D})e^{-\delta_{\mathcal{L}}(t-\underline{D})}\notag\\
	&+ \int_{\underline{D}}^t M_{\mathcal{L}}(t-\tau) e^{-\delta_{\mathcal{L}}(t-\tau)} \Gamma_2(\tau) d\tau:=\Gamma_{e,2}(t),\label{eq:game2}
\end{align}
where $M_{\mathcal{L}}(t)$, $\delta_{\mathcal{L}} > 0$ are obtained by applying the result of  Lemma \ref{lemerr1} to the Hurwitz  matrix $\mathcal{L}$.
Finally, for $t\ge t_f$, we have $\Delta U(t)=0$ from Lemma \ref{lemada2}, such that $\Delta e(t)=0$. Thus, the upper bound of the observer error is derived as
$
	|e_o(t)| \le \Gamma_{e,2}(t_f) M_{\mathcal{L}}(t-t_f)e^{-\delta_{\mathcal{L}}(t-t_f)}
$
 for $t\in[t_f,\infty)$. Recalling $\Omega_e$ defined in \eqref{omega},  we apply Young's inequality  to obtain the following bound $\Omega_e\leq\sqrt{2}|e_o(t)|$.

The continuation argument in Theorem \ref{theo2} extends these estimates to all $t\ge0$.
It readily follows from \eqref{oerrsol} that $e_o(t)$ is bounded with $\tilde{X}(0)\in\mathbb{R}^n$, $\tilde{V}_d(0)\in\mathbb{R}^{n_d}$, and bounded $\Delta e(t)$ within the finite time $t\in[0,t_f)$. Since the matrix $\mathcal{L}$ is Hurwitz, the error $e_o$ will be exponentially convergent to zero for $t\ge t_f$. The results of well-posedness and exponential convergence to zero are thus established.

\end{appendix}

\section*{References}
   \bibliographystyle{bibsty}                     
   \bibliography{biblio}

\end{document}